\documentclass[12pt]{article}%
\usepackage[cm]{fullpage}
\usepackage[usenames,dvipsnames,svgnames,table]{xcolor}%
\usepackage{bm}
\usepackage{paralist}%
\usepackage{picins}%
\usepackage{mleftright}%
\usepackage{stmaryrd}%
\usepackage{algorithm}%
\usepackage{algorithmic}%
\usepackage{esvect}
\usepackage{scalerel}%

\usepackage{xspace}%
\usepackage{ifpdf}%
\usepackage{graphicx}%
\usepackage[normalem]{ulem} %
\usepackage{caption}%
\usepackage{amsmath}%
\usepackage{amssymb}%
\usepackage[amsmath,thmmarks]{ntheorem}%
\theoremseparator{.}%

\usepackage{euscript}
\usepackage{mathrsfs}

\usepackage{verbatim}%

\usepackage{bbding}%

\usepackage{titlesec}%
\titlelabel{\thetitle. }%

\usepackage{bbm}%
\usepackage{bbding}

\usepackage{hyperref}%
\ifx\PDFBlackWhite\undefined%
\hypersetup{%
  breaklinks,%
  ocgcolorlinks, %
  colorlinks=true,%
  urlcolor=[rgb]{0.45,0.0,0.0},%
  linkcolor=[rgb]{0.0,0.45,0.0},%
  citecolor=[rgb]{0,0,0.45}%
}%
\fi

\allowdisplaybreaks

\newcommand{\etal}{\textit{et~al.}\xspace}
\newcommand{\eps}{\Mh{\varepsilon}}%

\newlength{\savedparindent}
\newcommand{\SaveIndent}{\setlength{\savedparindent}{\parindent}}

\newcommand{\Term}[1]{\textsf{#1}}

\definecolor{blue25}{rgb}{0, 0, 11}
\newcommand{\emphic}[2]{%
  \textcolor{blue25}{%
    \textbf{\emph{#1}}}%
  \index{#2}}

\ifx\PDFBlackWhite\undefined
\else
\renewcommand{\emphic}[2]{\textbf{\emph{#1}}}
\fi

\newcommand{\emphi}[1]{\emphic{#1}{#1}}

\newcommand{\cardin}[1]{\left| {#1} \right|}%

\newcommand{\pth}[1]{\mleft({#1}\mright)}
\newcommand{\brc}[1]{\left\{ {#1} \right\}}

\newcommand{\Set}[2]{\left\{ #1 \;\middle\vert\; #2 \right\}}
\newcommand{\pbrc}[1]{\mleft[ {#1} \mright]}

\newcommand{\Ex}[1]{\mathop{\mathbf{E}}\pbrc{#1}}

\newcommand{\Prob}[1]{\mathop{\mathbf{Pr}}\pbrc{#1}}

\newcommand{\remove}[1]{}

\newtheorem{theorem}{Theorem}[subsection]%
\newtheorem{lemma}[theorem]{Lemma}%
\newtheorem*{restate*}[theorem]{Restatement of }%
\newtheorem{corollary}[theorem]{Corollary}%

\newtheorem{observation}[theorem]{Observation}%

\newcommand{\myqedsymbol}{\rule{2mm}{2mm}}

\theoremstyle{remark}%
\theoremheaderfont{\sf}%
\theorembodyfont{\upshape}%
\newtheorem{defn}[theorem]{Definition}
\newtheorem{remark}[theorem]{Remark}%
\newtheorem{example}[theorem]{Example}%
\newtheorem*{remark:unnumbered}[theorem]{Remark}%

\theoremheaderfont{\em}%
\theorembodyfont{\upshape}%
\theoremstyle{nonumberplain}%
\theoremseparator{}%
\theoremsymbol{\myqedsymbol}%
\newtheorem{proof}{Proof:}%

\numberwithin{figure}{section}%
\numberwithin{table}{section}%
\numberwithin{equation}{section}%

\newcommand{\HLinkPageSuffix}[3]{\hyperref[#2]{#1\ref*{#2}%
      #3$_\text{p\pageref{#2}}$}}%
\newcommand{\HLinkSuffix}[3]{\hyperref[#2]{#1\ref*{#2}{#3}}}
\newcommand{\HLinkShort}[2]{\hyperref[#2]{#1\ref*{#2}}}
\newcommand{\HLink}[2]{\hyperref[#2]{#1~\ref*{#2}}}
\newcommand{\HLinkPage}[2]{\hyperref[#2]{#1~\ref*{#2}%
    $_\text{p\pageref{#2}}$}}

\newcommand{\eqlab}[1]{\label{equation:#1}}%
\newcommand{\Eqref}[1]{\HLinkSuffix{Eq.~(}{equation:#1}{)}}
\newcommand{\Eqrefpage}[1]{\HLinkPageSuffix{Eq.~(}{equation:#1}{)}}
\newcommand{\eqrefpar}[1]{\hyperref[equation:#1]{(\ref*{equation:#1})}} %

\newcommand{\seclab}[1]{\label{sec:#1}} %
\newcommand{\secref}[1]{\HLink{Section}{sec:#1}} %

\providecommand{\deflab}[1]{\label{def:#1}}
\newcommand{\defref}[1]{\HLink{Definition}{def:#1}}

\newcommand{\apndlab}[1]{\label{apnd:#1}}
\newcommand{\apndref}[1]{\HLink{Appendix}{apnd:#1}}

\newcommand{\exmlab}[1]{\label{example:#1}}
\newcommand{\exmref}[1]{\HLink{Example}{example:#1}}

\newcommand{\lemlab}[1]{\label{lemma:#1}}
\newcommand{\lemref}[1]{\HLink{Lemma}{lemma:#1}}

\newcommand{\remlab}[1]{\label{rem:#1}}
\newcommand{\remref}[1]{\HLink{Remark}{rem:#1}}

\newcommand{\obslab}[1]{\label{observation:#1}}
\newcommand{\obsref}[1]{\HLink{Observation}{observation:#1}}

\newcommand{\thmlab}[1]{{\label{theo:#1}}}
\newcommand{\thmref}[1]{\HLink{Theorem}{theo:#1}}

\newcommand{\pr}{\Mh{\tau}}%

\providecommand{\Mh}[1]{{#1}}

\renewcommand{\th}{th\xspace}

\newcommand{\ds}{\displaystyle}%
\renewcommand{\Re}{{\mathbb{R}}}
\newcommand{\ANN}{\Term{ANN}\xspace}%
\newcommand{\NN}{\Term{NN}\xspace}%

\newcommand{\PntSet}{\ensuremath{\Mh{P}}\xspace}%

\newcommand{\Hd}{\Mh{\mathbb{H}^d}}%
\newcommand{\N}{\Mh{\mathcal{N}}}%
\newcommand{\Family}{\Mh{\mathcal{H}}}%

\newcommand{\colY}[2]{\Mh{\mathcal{C}}_{#2}\pth{#1}}%

\newcommand{\Ni}[1]{\Mh{\mathsf{s}}\pth{#1}}
\newcommand{\cpi}[1]{\Mh{\mathsf{f}}_{#1}}

\newcommand{\sCi}[1]{\Mh{\gamma}_{#1}}
\newcommand{\nCollX}[1]{\Mh{\Gamma}_{#1}}

\newcommand{\distHC}{\mathrm{d}_H}
\newcommand{\distH}[2]{\mathrm{d}_H\pth{#1, #2}}

\newcommand{\Pclose}{\PntSet_{\leq}}
\newcommand{\Pfar}{\PntSet_{>}}
\newcommand{\constA}{\Mh{\mathsf{c}}}%
\newcommand{\constB}{\Mh{\mathsf{c}_1}}%

\newcommand{\sortX}[1]{\mathrm{sort}\pth{#1}}%
\newcommand{\tailY}[2]{\Mh{#1^{}_{\bbslash #2}}}%

\newcommand{\cCoord}{\Mh{\alpha}}%
\newcommand{\cTimes}{\Mh{\beta}}%
\newcommand{\cDSTimes}{\Mh{\delta}}%
\newcommand{\nnConst}{\Mh{c}}%
\newcommand{\p}{\Mh{\rho}}%
\newcommand{\Lp}{\Mh{L}_{\p}}

\newcommand{\ts}{\hspace{0.6pt}}
\newcommand{\DA}{\Mh{D}}%
\newcommand{\DSTimes}{\Mh{L}}%

\newcommand{\rr}{\Mh{r}}%

\newcommand{\RR}{\Mh{R}}%

\newcommand{\RRA}{\Mh{{U}}}%

\DeclareMathAlphabet{\mathpzc}{OT1}{pzc}{m}{it}

\newcommand{\subseq}{\Mh{{\mathpzc{m}}}}%

\newcommand{\subseqA}{\Mh{{\mathpzc{n}}}}%

\newcommand{\seq}{\bm{\Mh{M}}}%

\newcommand{\pnt}{\Mh{\bm{x}}}%
\newcommand{\pntc}{\Mh{{x}}}%
\newcommand{\nnpnt}{\Mh{\bm{n}}}%
\newcommand{\rmC}[2]{{#1}^{}_{\setminus #2}}

\newcommand{\truncY}[2]{{#1}_{\leq #2}}%

\newcommand{\pntA}{\Mh{\bm{v}}}%
\newcommand{\pntAc}{\Mh{{v}}}%

\newcommand{\pntB}{\Mh{\bm{z}}}%
\newcommand{\nnfold}[2]{\mathsf{nn}^{}_{\bbslash #1}\pth{ #2}}%

\newcommand{\tTimes}{\Mh{t}}%

\newcommand{\query}{\Mh{\bm{q}}}%

\newcommand{\threshold}{\Mh{\psi}}

\newcommand{\dist}{\mathrm{dist}}

\newcommand{\snorm}[2]{\left\| \smash{#2} \right\|_{#1}}
\newcommand{\norm}[2]{\left\| {#2} \right\|_{#1}}

\newcommand{\normT}[3]{\norm{#1}{#3_{\leq #2}}}%
\newcommand{\Var}[1]{\mathop{\mathbf{V}}\pbrc{#1}}

\newcommand{\Event}{\Mh{\mathcal{E}}}%

\newcommand{\atgen}{\symbol{'100}}

\newcommand{\SarielThanks}[1]{%
   \thanks{%
      Department of Computer Science; %
      University of Illinois; %
      201 N. Goodwin Avenue; %
      Urbana, IL, %
      61801, USA; %
      {\tt \si{sariel}\atgen{}\si{illinois.edu}}; %
      {\tt \url{http://sarielhp.org}.} #1 } 
}

\newcommand{\SepidehThanks}[1]{%
   \thanks{%
      Department of EECS; 
      MIT; %
      77 Massachusetts Avenue, Cambridge, MA 02139, USA;
      {\tt mahabadi@mit.edu}. %
      #1 
   }%
}

\newcommand{\si}[1]{#1}%
\newcommand{\IntRange}[1]{\left\llbracket #1 \right\rrbracket}
\newcommand{\DistD}[1]{\Mh{\mathcal{D}}_{#1}}

\newcommand{\nfrac}[2]{#1/#2}%

\newcommand{\tldO}{\scalerel*{\widetilde{O}}{j^2}}%

\newcommand{\val}{\Mh{\nu}}%
\newcommand{\trc}{\Mh{\mathrm{trc}}}%
\newcommand{\trcZ}[3]{%
   \trc^{}_{#1, #2}\pth{ #3}%
}

\newcommand{\DistB}[1]{\Mh{\EuScript{B}}_{#1}} 
\newcommand{\DistBY}[2]{\Mh{\DistB{\vecC/#1}^{#2}}}

\newcommand{\NNV}[1]{#1^{\Mh{*}}}%
\newcommand{\nnq}{\NNV{\query}}%

\newcommand{\vecW}{\Mh{\bm{w}}} 

\newcommand{\AdmS}{\Mh{\mathcal{G}}}%

\newcommand{\wc}{\Mh{w}}

\newcommand{\nnNotation}{\Mh{\mathsf{nn}}}%
\newcommand{\nnCZ}[3]{\nnNotation^{}_{#1}\pth{#2, #3}}%
\newcommand{\nnCX}[1]{\nnNotation\pth{#1}}%

\DefineNamedColor{named}{OliveGreen}{cmyk}{0.64, 0, 0.95, 0.40}
\definecolor{OliveGreen}{cmyk}{0.64, 0, 0.95, 0.40 }

\providecommand{\ProblemC}[1]{\textsf{#1}}%

\providecommand{\ComplexityClass}[1]{{{{%
            \textsc{#1}}}}}
\providecommand{\NPHard}{\ComplexityClass{NP-Hard}\xspace}
\providecommand{\NPComplete}{\ComplexityClass{NP-Complete}\xspace}

\newcommand{\PTAS}{\Term{PTAS}\xspace}%
\newcommand{\LSH}{\Term{LSH}\xspace}%

\newcommand{\dnnZ}[3]{\Mh{\mathsf{d}_{#1}\pth{#2,#3}}}%
\newcommand{\daZ}[3]{\Mh{\mathsf{d}_{#1}\pth{#2,#3}}}%

\newcommand{\CSet}{\Mh{I}}

\newcommand{\DS}{\Term{DS}\xspace}%

\newcommand{\badCoords}{\Mh{B}}
\newcommand{\goodCoords}{\Mh{G}}%
\newcommand{\cc}{\Mh{\xi}}%
\newcommand{\vecC}{\Mh{\bm{\xi}}}%
\IfFileExists{sariel_computer.sty}{\def\sarielComp{1}}{}
\ifx\sarielComp\undefined%
\newcommand{\SarielComp}[1]{}
\newcommand{\NotSarielComp}[1]{#1}%
 \else
\newcommand{\SarielComp}[1]{#1}%
\newcommand{\NotSarielComp}[1]{}%
\fi

\SarielComp{%
   \ifx\colorMath\undefined%
   \else
   \DefineNamedColor{named}{ColorMath}{rgb}{0.5,0.2,0}
       \renewcommand{\Mh}[1]{{\textcolor{ColorMath}{#1}}}%
   \fi
}

\title{Proximity in the Age of Distraction:\\%
   Robust Approximate Nearest Neighbor Search}

\author{%
   Sariel Har-Peled%
   \SarielThanks{%
      Work on this paper was partially supported by a NSF AF awards
      CCF-0915984 and CCF-1217462.}%
   \and%
   Sepideh Mahabadi%
   \SepidehThanks{%
      This work was supported in part by NSF grant \si{IIS} -1447476.%
   }%
}

\begin{document}
\clearpage%
\maketitle

\begin{abstract}
    We introduce a new variant of the nearest neighbor search problem,
    which allows for some coordinates of the dataset to be arbitrarily
    corrupted or unknown.  Formally, given a dataset of $n$ points
    $\PntSet=\brc{ \pnt_1,\ldots, \pnt_n}$ in high-dimensions, and a
    parameter $k$, the goal is to preprocess the dataset, such that
    given a query point $\query$, one can compute quickly a point
    $\pnt \in \PntSet$, such that the distance of the query to the
    point $\pnt$ is minimized, when ignoring the ``optimal'' $k$
    coordinates. Note, that the coordinates being ignored are a
    function of both the query point and the point returned.

    We present a general reduction from this problem to answering \ANN
    queries, which is similar in spirit to \LSH (locality sensitive
    hashing) \cite{im-anntr-98}.  Specifically, we give a sampling
    technique which achieves a bi-criterion approximation for this
    problem. If the distance to the nearest neighbor after ignoring
    $k$ coordinates is $\rr$, the data-structure returns a point that
    is within a distance of $O(\rr)$ after ignoring $O(k)$
    coordinates.  We also present other applications and further
    extensions and refinements of the above result.

    The new data-structures are simple and (arguably) elegant, and
    should be practical -- specifically, all bounds are polynomial in
    all relevant parameters (including the dimension of the space, and
    the robustness parameter $k$).
\end{abstract}


\section{Introduction}

The \emph{nearest neighbor} problem (\NN) is a fundamental geometric
problem which has major applications in many areas such as databases,
computer vision, pattern recognition, information retrieval, and many
others. Given a set $\PntSet$ of $n$ points in a $d$-dimensional
space, the goal is to build a data-structure, such that given a query
point $\query$, the algorithm can report the closest point in
$\PntSet$ to the query $\query$. A particularly interesting and
well-studied instance is when the points live in a $d$-dimensional
real vector space $\Re^d$.  Efficient exact and approximate algorithms
are known for this problem.  (In the \emph{approximate nearest
   neighbor} (\ANN) problem, the algorithm is allowed to report a
point whose distance to the query is larger by at most a factor of
$(1+\eps)$, than the real distance to the nearest point.)  See
\cite{amnsw-oaann-98, k-tanns-97, him-anntr-12, kor-esann-00,
   h-rvdnl-01, kl-nnsap-04, diim-lshsb-04, cr-orcpl-10, p-ebnns-06,
   ac-annfjl-06, ai-nohaa-08, ainr-blsh-14, ar-oddha-15}, the surveys
\cite{s-fmmds-05}, \cite{i-nnhds-04} and \cite{sdi-nnmlv-06}, and
references therein (of course, this list is by no means exhaustive).
One of the state of the art algorithms for \ANN, based on Locality
Sensitive Hashing, finds the $(1+\eps)$-\ANN with query time
$\tldO(dn^{c})$, and preprocessing/space $O(dn+ n^{1+c})$ in $L_1$,
where $c=O((1+\eps)^{-1})$ \cite{im-anntr-98}, where $\tldO$ hides a
constant that is polynomial in $\log n$. For the $L_2$ norm, this
improves to $c = O({ (1+\eps)^{-2}})$ \cite{ai-nohaa-08}.

Despite the ubiquity of nearest neighbor methods, the vast majority of
current algorithms suffer from significant limitations when applied to
data sets with corrupt, noisy, irrelevant or incomplete data.  This is
unfortunate since in the real world, rarely one can acquire data
without some noise embedded in it. This could be because the data is
based on real world measurements, which are inherently noisy, or the
data describe complicated entities and properties that might be
irrelevant for the task at hand.

In this paper, we address this issue by formulating and solving a
variant of the nearest neighbor problem that allows for some data
coordinates to be arbitrarily corrupt.  Given a parameter $k$, the
\emphi{$k$-Robust Nearest Neighbor} for a query point $\query$, is a
point $\pnt \in \PntSet$ whose distance to the query point is
minimized ignoring ``the optimal'' set of $k$-coordinates (the term
`robust' \ANN is used as an analogy to \emph{Robust PCA}
\cite{clmw-rpca-11}).  That is, the $k$ coordinates are chosen so that
deleting these coordinates, from both $\pnt$ and $\query$ minimizes
the distance between them. In other words, the problem is to solve the
\ANN problem in a different space (which is definitely not a metric),
where the distance between any two points is computed ignoring the
worst $k$ coordinates.  To the best of our knowledge, this is the
first paper considering this formulation of the Robust \ANN problem.

This problem has natural applications in various fields such as
computer vision, information retrieval, etc.  In these applications,
the value of some of the coordinates (either in the dataset points or
the query point) might be either corrupted, unknown, or simply
irrelevant.  In computer vision, examples include image de-noising
where some percent of the pixels are corrupted, or image retrieval
under partial occlusion (e.g. see \cite{he-scump-07}), where some part
of the query or the dataset image has been occluded.  In these
applications there exists a perfect match for the query after we
ignore some dimensions.  Also, in medical data and recommender
systems, due to incomplete data \cite{swcsr-mimde-09,
   cfvrs-mdmdi-13, wcnk-shmde-13}, not all the features (coordinates)
are known for all the people/recommendations (points), and moreover,
the set of known values differ for each point.  Hence, the goal is to
find the perfect match for the query ignoring some of those features.

For the binary hypercube, under the Hamming distance, the $k$-robust
nearest neighbor problem is equivalent to the \emph{near neighbor}
problem. The near neighbor problem is the decision version of the \ANN
search, where a radius $\rr$ is given to the data structure in
advance, and the goal is to report a point that is within distance
$\rr$ of the query point. Indeed, there exists a point $\pnt$ within
distance $\rr$ of the query point $\query$ if and only if $\rr$
coordinates can be ignored such that the distance between $\pnt$ and
$\query$ is zero.

\paragraph{Budgeted Version.}

We also consider the weighted generalization of the problem where the
amount of uncertainty varies for each feature. In this model, each
coordinate is assigned a weight $0\leq w_i\leq 1$ in advance, which
tries to capture the certainty level about the value of the coordinate
($w_i=1$ indicates that the value of the coordinate is correct and
$w_i=0$ indicates that it can not be trusted). The goal is to ignore a
set of coordinates $\badCoords$ of total weight at most $1$, and find
a point $\pnt \in \PntSet$, such that the distance of the query to the
point $\pnt$ is minimized ignoring the coordinates in
$\badCoords$. Surprisingly, even computing the distance between two
points under this measure is \NPComplete (it is almost an instance of
\ProblemC{Min Knapsack}).

\SaveIndent%


\subsection{Results}
We present the following new results:
\begin{compactenum}[(A)]%
    \item \textsc{Reduction from Robust \ANN to \ANN.}  %
    We present a general reduction from the robust \ANN problem to the
    ``standard'' \ANN problem.  This results in a bi-criterion
    constant factor approximation, with sublinear query time, for the
    $k$-robust nearest neighbor problem.
    
    \begin{compactenum}[(I)]
        \item For {$L_1$-norm} the result can be stated as follows. If
        there exists a point $\nnq$ whose distance to the query point
        $\query$ is at most $\rr$ by ignoring $k$ coordinates, the new
        algorithm would report a point $\pnt$ whose distance to the
        query point is at most $O(r/\delta)$, ignoring $O(k/\delta)$
        coordinates. The query algorithm performs $O(n^{\delta})$ \ANN
        queries in $2$-\ANN data structures, where $\delta \in (0,1)$
        is a prespecified parameter.

        \item In \secref{sec:lp:case}, we present the above result in
        the somewhat more general settings of the $\Lp$ norm.  The
        algorithm reports a point whose distance is within
        $O(r(c+\frac{1}{\delta})^{1/\rho})$ after ignoring
        $O(k( \frac{1}{\cDSTimes}+\nnConst))$ coordinates while
        performing $n^{\delta}$ of $c^{1/\p}$-\ANN queries. 
    \end{compactenum}

    \smallskip%
    \item \textsc{$(1+\eps)$-approximation.}  %
    We modify the new algorithm to report a point whose distance to
    the query point is within $\rr(1+\eps)$ by ignoring
    $O( \frac{k}{\delta\eps})$ coordinates while performing
    $\tldO(n^{\delta}/\eps)$ \ANN queries (specifically,
    $(1+O(\eps))$-\ANN). For the sake of simplicity of exposition, we
    present this extension only in the $L_1$ norm. See
    \apndref{eps:approx} for details.

    \smallskip%
    \item \textsc{Budgeted version.}  %
    In \secref{sec:budgeted}, we generalize our algorithm for the
    weighted case of the problem. If there exists a point within
    distance $\rr$ of the query point by ignoring a set of coordinates
    of weight at most $1$, then our algorithm would report a point
    whose distance to the query is at most $O(\rr)$ by ignoring a set
    of coordinates of weight at most $O(1)$. Again, for the sake of
    simplicity of exposition, we present this extension only in the
    $L_1$ norm.

    \smallskip%
    \item \textsc{Data sensitive L{S}H queries.} %
    It is a well known phenomenon in proximity search (e.g. see Andoni
    \etal~\cite[Section 4.5]{adiim-lshus-06}), that many
    data-structures perform dramatically better than their theoretical
    analysis. Not only that, but also they find the real nearest
    neighbor early in the search process, and then spend quite a bit
    of time on proving that this point is indeed a good \ANN. It is
    natural then to ask whether one can modify proximity search
    data-structures to take an advantage of such a behavior. That is,
    if the query is easy, the data-structure should answer quickly
    (and maybe even provide the exact nearest neighbor in such a
    case), but if the instance is hard, then the data-structure works
    significantly harder to answer the query.
    
    As an application of our sampling approach, we show a
    data-sensitive algorithm for \LSH for the binary hypercube case
    under the Hamming distance. The new algorithm solves the
    Approximate Near Neighbor problem, in time
    $O\pth{ d \exp(\Delta) \log n }$, where $\Delta$ is the smallest
    value with 
    $ \sum_{i=1}^n \exp \pth{\bigl.-\Delta\,\dist(\query,\pnt_i) /
       \rr} \leq 1, $
    where $\dist(\query,\pnt_i)$ is the distance of the query $\query$
    to the $i$\th point $\pnt_i \in \PntSet$, and $\rr$ is the
    distance being tested.

    We also get that such \LSH queries works quickly on low
    dimensional data, see \remref{l:dim} for details.
\end{compactenum}%
\smallskip%
The new algorithms are clean and should be practical.  Moreover, our
results for the $k$-robust \ANN hold for a wide range of the parameter
$k$, from $O(1)$ to $O(d)$.


\subsection{Related Work}
There has been a large body of research focused on adapting widely
used methods for high-dimensional data processing to make them
applicable to corrupt or irrelevant data. For example, Robust PCA
\cite{clmw-rpca-11} is an adaptation of the PCA algorithm that handles
a limited amount of adversarial errors in the input matrix. Although
similar in spirit, those approaches follow a different technical
development than the one in this paper.

Also, similar approaches to robustness has been used in theoretical
works. In the work of Indyk on $\Lp$ sketching \cite{i-sdpge-06}, the
distance between two points $x$ and $y$, is defined to be the median
of $\{|x_1-y_1|,....,|x_d-y_d|\}$. Thus, it is required to compute the
$L_{\infty}$ norm of their distance, but only over the smallest $d/2$
coordinates.

Finally, several generalizations of the \ANN problem have been
considered in the literature. Two related generalizations are the
Nearest $k$-flat Search and the Approximate Nearest $k$-flat. In the
former, the dataset is a set of $k$-flats ($k$-dimensional affine
subspace) instead of simple points but the query is still a point (see
\cite{m-anlsh-15} for example). In the latter however, the dataset
consists of a set of points but the query is now a $k$-flat (see for
example \cite{aikn-alnnh-09, mnss-akfnn-14}). We note that our problem
cannot be solved using these variations (at least naively) since the
set of coordinates that are being ignored in our problem are not
specified in advance and varies for each query. This would mean that
$d \choose k$ different subspaces are to be considered for each point.
In our settings, both $d$ and $k$ can be quite large, and the new
data-structures have polynomial dependency in both parameters.

\paragraph{Data sensitive \LSH.}

The fast query time for low dimensional data was demonstrated before
for an \LSH scheme \cite[Appendix A]{diim-lshsb-04} (in our case, this
is an easy consequence of our data-structure).  Similarly, optimizing
the parameters of the \LSH construction to the hardness of the
data/queries was done before \cite[Section 4.3.1]{adiim-lshus-06} --
our result however does this on the fly for the query, depending on
the query itself, instead of doing this fine tuning of the whole
data-structure in advance for all queries.

\subsection{Techniques}

By definition of the problem, we cannot directly apply
Johnson-Lindenstrauss lemma to reduce the dimensions (in the $L_2$
norm case). Intuitively, dimension reduction has the reverse effect of
what we want -- it spreads the mass of a coordinate ``uniformly'' in
the projection's coordinates -- thus contaminating all projected
coordinates with noise from the ``bad'' coordinates.

The basic idea of our approach is to generate a set of random
projections, such that all of these random projections map far points
to far points (from the query point), and at least one of them
projects a close point to a close point. Thus, doing \ANN queries in
each of these projected point sets, generates a set of candidate
points, one of them is the desired robust \ANN.

Our basic approach is based on a simple sampling scheme, similar to
the Clarkson-Shor technique \cite{cs-arscg-89} and \LSH
\cite{him-anntr-12}.  The projection matrices we use are
\emph{probing} matrices. Every row contains a single non-zero entry,
thus every row copies one original coordinate, and potentially scales
it up by some constant.
  
\paragraph{A sketch of the technique:} %

Consider the case where we allow to drop $k$ coordinates.  For a given
query point $\query$, it has a robust nearest neighbor
$\nnq \in \PntSet$, such that there is a set $\badCoords$ of $k$
``bad'' coordinates, such that the distance between $\query$ and
$\nnq$ is minimum if we ignore the $k$ coordinates of $\badCoords$
(and this is minimum among all such choices).

We generate a projection matrix by picking the $i$\th coordinate to be
present with probability $1/(\cCoord k)$, where $\cCoord$ is some
constant, for $i=1,\ldots, d$.  Clearly, the probability that such a
projection matrix avoids picking the $k$ bad coordinates is
$(1-\frac{1}{\cCoord k})^k \approx e^{-1/\cCoord}$. In particular, if
we repeat this process $\cTimes \ln n$ times, where $\cTimes$ is some
constant, then the resulting projection avoids picking any bad
coordinate with probability
$\approx e^{-\cTimes\ln n/\cCoord} = n^{-\cTimes/\cCoord}$. On the
other hand, imagine a ``bad'' point $\pnt \in \PntSet$, such that one
has to remove, say, $(\cCoord /\cTimes) k$ coordinates before the
distance of the point to the query $\query$ is closer than the robust
\NN $\nnq$ (when ignoring only $k$ coordinates). Furthermore, imagine
the case where picking any of these coordinates is fatal -- the value
in each one of these bad coordinates is so large, that choosing any of
these bad coordinates results in this bad point being mapped to a far
away point. Then, the probability that the projection fails to select
any of these bad coordinates is going to be roughly
\begin{math}
    (1- \frac{1}{\cCoord k})^{\cCoord k \ln n} \approx 1/n.
\end{math}
Namely, somewhat informally, with decent probability all bad points
get mapped to faraway points, and the near point gets mapped to a
nearby point. Thus, with probability roughly
$\approx n^{-\cTimes/\cCoord}$, doing a regular \ANN query on the
projected points, would return the desired \ANN. As such, repeating
this embedding $O(n^{\cTimes/\cCoord} \log n)$ times, and returning
the best \ANN encountered would return the desired robust \ANN with
high probability.

\paragraph{The good, the bad, and the truncated.}

Ultimately, our technique works by probing the coordinates, trying to
detect the ``hidden'' mass of the distance of a bad point from the
query. The mass of such a distance might be concentrated in few
coordinates (say, a point has $k+1$ coordinates with huge value in
them, but all other coordinates are equal to the query point) -- such
a point is arguably still relatively good, since ignoring slightly
more than the threshold $k$ coordinates results in a point that is
pretty close by.

On the other hand, a point where one has to ignore a large number of
coordinates (say $2k$) before it becomes reasonably close to the query
point is clearly bad in the sense of robustness. As such, our
data-structure would classify points, where one has to ignore slightly
more than $k$ coordinates to get a small distance, as being close.

To capture this intuition, we want to bound the influence a single
coordinate has on the overall distance between the query and a
point. To this end, if the robust nearest neighbor distance, to
$\query$ when ignoring $k$ coordinates, is $\rr$, then we consider
capping the contribution of every coordinate, in the distance
computation, by a certain value, roughly, $\rr/k$. Under this
\emph{truncation}, our data-structure returns a point that is $O(\rr)$
away from the query point, where $\rr$ is the distance to the
$k$-robust \NN point.

Thus, our algorithm can be viewed as a bicriterion approximation
algorithm - it returns a point where one might have to ignore slightly
more coordinates than $k$, but the resulting point is constant
approximation to the nearest-neighbor when ignoring $k$ points.

In particular, a point that is still bad after such an aggressive
truncation, is amenable to the above random probing. By carefully
analyzing the variance of the resulting projections for such points,
we can prove that such points would be rejected by the \ANN
data-structure on the projected points.

\paragraph{Budgeted version.}

To solve the budgeted version of the problem we use a similar
technique to importance sampling. If the weight of a coordinate is
$w_i$, then in the projection matrix, we sample the $i$\th coordinate
with probability $\approx 1/w_i$ and scale it by a factor of
$\approx w_i$. This would ensure that the set of ``bad" coordinates
are not sampled with probability $e^{-2/\cCoord}$. Again we repeat it
$\cTimes\ln n$ times to get the desired bounds.

\paragraph{Data sensitive \LSH.}

The idea behind the data-sensitive \LSH, is that \LSH can be
interpreted as an estimator of the local density of the point set. In
particular, if the data set is sparse near the query point, not only
the \LSH data-structure would hit the nearest-neighbor point quickly
(assuming we are working in the right resolution), but furthermore,
the density estimation would tell us that this event happened. As
such, we can do the regular exponential search -- start with an
insensitive \LSH scheme (that is fast), and go on to use more
sensitive \LSH{}s, till the density estimation part tells us that we
are done. Of course, if all fails, the last \LSH data-structure used
is essentially the old \LSH scheme.

\section{Preliminaries}

\subsection{The problem}

\begin{defn}
    \deflab{tail}%
    For a point $\pnt \in \Re^d$, let $\pi = \sortX{\pnt}$ be a
    permutation of $\IntRange{d} = \brc{1,\ldots, d}$, such that
    \begin{math}
        |\pntc_{\pi(1)} | %
        \geq%
        |\pntc_{\pi(2)} | %
        \geq%
        \cdots%
        \geq%
        |\pntc_{\pi(d)} |.
    \end{math}
    For a parameter $1 \leq i \leq d$, the \emphi{$i$-tail} of $\pnt$
    is the point
    \begin{math}
        \tailY{\pnt}{i}
        =%
        \bigl( 0, \ldots, 0, |\pntc_{\pi(i+1)} |, \allowbreak
        |\pntc_{\pi(i+2)} |, \ldots, |\pntc_{\pi(d)} | \bigr).
    \end{math}
    Note, that given $\pnt \in \Re^d$ and $i$, computing
    $\tailY{\pnt}{i}$ can be done in $O(d)$ time, by median selection.
\end{defn}

Thus, given two points $\pnt, \pntA \in \Re^d$, their distance (in the
$\Lp$-norm), ignoring the $k$ worst coordinates (which we believe to
be noise), is $\norm{\p}{\tailY{\pth{u-v}}{k}}$.  Here, we are
interested in computing the nearest neighbor when one is allowed to
ignore $k$ coordinates. Formally, we have the following:

\begin{defn}
    \deflab{k:fold:n:n}%
    For parameters $\p$, $k$, a set of points
    $\PntSet \subseteq \Re^d$, and a query point $\query \in \Re^d$,
    the \emphi{$k$-robust nearest-neighbor} to $\query$ in $\PntSet$
    is
    \begin{math}%
        \ds%
        \nnfold{k}{\PntSet,\query}%
        =%
        \arg \min_{\pnt \in \PntSet}
        \norm{\p}{\smash{\tailY{\pth{\pnt-\query} }{k}} \bigr.},
    \end{math}
    which can be computed, naively, in $O( d \cardin{\PntSet})$ time.
\end{defn}

\begin{defn}
    \deflab{r:m:c}%
    For a point $\pnt\in \Re^d$ and a set of coordinates
    $\CSet \subseteq \IntRange{d}$, we define $\rmC{\pnt}{\CSet}$ to
    be a point in $\Re^{d-\cardin{\CSet}}$ which is obtained from
    $\pnt$ by deleting the coordinates that are in $\CSet$.
\end{defn}


\subsubsection{Projecting a point set}



\begin{defn}
    \deflab{set:proj}
    Consider a sequence $\subseq$ of $\ell$, not necessarily distinct,
    integers $i_1, i_2, \ldots, i_\ell \in \IntRange{d}$, where
    $\IntRange{d} = \brc{1, \ldots, d}$.  For a point
    $\pnt = (\pntc_1, \ldots, \pntc_d) \in \Re^d$, its
    \emphi{projection} by $\subseq$, denoted by $\subseq \pnt$, is the
    point $\pth{ \pntc_{i_1}, \ldots, \pntc_{i_\ell}} \in \Re^\ell$.
    Similarly, the \emph{projection} of a point set
    $\PntSet \subseteq \Re^d$ by $\subseq$ is the point set
    $\subseq \PntSet = \Set{\subseq \pnt}{\pnt \in \PntSet}$.
    Given a weight vector $\vecW = \pth{\wc_1, \ldots, \wc_d}$ the
    \emphi{weighted projection} by a sequence $\subseq$ of a point
    $\pnt$ is
    $\subseq^{}_{\vecW}( \pnt) = \pth{ \wc_{i_1} \pntc_{i_1}, \ldots,
       \wc_{i_\ell} \pntc_{i_\ell}} \in \Re^\ell$.

    Note, that can interpret $\subseq$ as matrix with dimensions
    $\ell \times d$, where the $j$\th row is zero everywhere except
    for having $\wc_{i_j}$ at the $i_j$\th coordinate (or $1$ in the
    unweighted case), for $j=1, \ldots,\ell$. This is a restricted
    type of a projection matrix.
\end{defn}

\begin{defn}
    \deflab{t:splay}%
    Given a probability $\pr > 0$, a natural way to create a
    projection, as defined above, is to include the $i$\th coordinate,
    for $i=1,\ldots, d$, with probability $\pr$. Let $\DistD{\pr}$
    denote the distribution of such projections.
    
    Given two sequences $\subseq = i_1, \ldots, i_\ell$ and
    $\subseqA = j_1, \ldots, j_{\ell'}$, let $\subseq | \subseqA$
    denote the \emph{concatenated} sequence
    \begin{math}
        \subseq | \subseqA = i_1, \ldots, i_\ell, j_1, \ldots,
        j_{\ell'}.
    \end{math}
    Let $\DistD{\pr}^\tTimes$ denote the distribution resulting from
    concatenating $t$ such independent sequences sampled from
    $\DistD{\pr}$.  (I.e., we get a random projection matrix, which is
    the result of concatenating $\tTimes$ independent projections.)
\end{defn}

Observe that for a point $\pnt \in \Re^d$ and a projection
$\seq \in \DistD{\pr}^\tTimes$, the projected point $\seq \pnt$ might
be higher dimensional than the original point $\pnt$ as it might
contain repeated coordinates of the original point.

\begin{remark}[{Compressing the projections}]
    \remlab{compress}%
    Consider a projection $\seq \in \DistD{\pr}^\tTimes$ that was
    generated by the above process (for either the weighted or
    unweighted case). Note, that since we do not care about the order
    of the projected coordinates, one can encode $\seq$ by counting
    for each coordinate $i \in \IntRange{d}$, how many time it is
    being projected. As such, even if the range dimension of $\seq$ is
    larger than $d$, one can compute the projection of a point in
    $O(d)$ time. One can also compute the distance between two such
    projected points in $O(d)$ times.
\end{remark}


\section{The $k$-robust \ANN under the $\Lp$-norm}
\seclab{sec:lp:case}%

In this section, we present an algorithm for approximating the
$k$-robust nearest neighbor under the $\Lp$-norm, where $\p$ is some
prespecified fixed constant (say $1$ or $2$). As usual in such
applications, we approximate the $\p$\th power of the $\Lp$-norm,
which is a sum of $\p$\th powers of the coordinates.

\subsection{The preprocessing and query algorithms}
\seclab{sec:lp:alg} 

\paragraph{Input.}
The input is a set $\PntSet$ of $n$ points in $\Re^d$, and a parameter
$k$.  Furthermore, we assume we have access to a data-structure that
answer (regular) $\nnConst^{1/\p}$-\ANN queries efficiently, where
$\nnConst$ is a quality parameter associated with these
data-structures.

\paragraph{Preprocessing.}  %
Let $\cCoord, \cTimes, \cDSTimes$ be three constants to be specified
shortly, such that $\cDSTimes \in (0,1)$.  We set
$\pr = 1/(\cCoord k)$, $\tTimes = \cTimes \ln n$, and
$\DSTimes = O( n^{\cDSTimes} \log n)$. We randomly and independently
pick $\DSTimes$ sequences
\begin{math}
    \seq_1, \ldots, \seq_\DSTimes \in \DistD{\pr}^\tTimes.
\end{math}
Next, the algorithm computes the point sets
$\PntSet_i = \seq_i \PntSet$, for $i=1,\ldots, \DSTimes$, and
preprocesses each one of them for $\nnConst^{1/\p}$-approximate
nearest-neighbor queries for the $\Lp$-norm ($\nnConst\geq 1$), using
a standard data-structure for \ANN that supports this.  Let $\DA_i$
denote the resulting \ANN data-structure for $\PntSet_i$, for
$i=1,\ldots, \DSTimes$ (for example we can use the data structure of
Indyk and Motwani \cite{im-anntr-98,him-anntr-12} for the $L_1/L_2$
cases).

\paragraph{Query.} %
Given a query point $\query \in \Re^d$, for $i=1, \ldots, \DSTimes$,
the algorithm computes the point $\query_i = \seq_i \query$, and its
\ANN $\pntA_i$ in $\PntSet_i$ using the data-structure $\DA_i$. Each
computed point $\pntA_i\in \PntSet_i$ corresponds to an original point
$\pnt_i \in \PntSet$. The algorithm returns the $k$-robust nearest
neighbor to $\query$ (under the $\Lp$-norm) among
$\pnt_1, \ldots, \pnt_\DSTimes$ via direct calculation.



\subsection{Analysis}

\subsubsection{Points: Truncated, light and heavy}
\begin{defn}
    \deflab{truncate}%
    For a point $\pnt = (\pntc_1, \ldots, \pntc_d) \in \Re^d$, and a
    threshold $\threshold > 0$, let
    $\truncY{\pnt}{\threshold} = \pth{\pntc_1',\ldots, \pntc_d'}$ be
    the \emphi{$\threshold$-truncated} point, where
    $\pntc_i' = \min\pth{ \cardin{\pntc_i}, \threshold}$, for
    $i=1,\ldots,d$. In words, we max out every coordinate of $\pnt$ by
    the threshold $\threshold$.  As such, the
    \emph{$\threshold$-truncated $\Lp$-norm} of
    $\pnt=(\pntc_1,\cdots, \pntc_d)$ is
    \begin{math}
        \normT{\p}{\threshold}{\pnt}%
        =%
        \pth{ \sum\nolimits_{i=1}^d \pth{
              \bigl.\min\pth{|\pntc_i|,\threshold}}^\p}^{1/\p}.
    \end{math}
\end{defn}

\begin{defn}
    For parameters $\threshold$ and $\rr$, a point $\pnt$ is
    \emphi{$(\threshold, \rr)$-light} if
    $\normT{\p}{\threshold}{\pnt} \leq \rr$. Similarly, for a
    parameter $\RR> 0$, a point is
    \emphi{$(\threshold, \RR)$-heavy} if
    $\normT{\p}{\threshold}{\pnt} \geq \RR$.
\end{defn}

Intuitively a light point can have only a few large coordinates.  The
following lemma shows that being light implies a small tail.

\begin{lemma}%
    \lemlab{lp:light} %
    For a number $\rr$, if a point $\pnt \in \Re^d$ is
    $( \nfrac{\rr}{k^{1/\p}}, \rr)$-light, then
    $\norm{\p}{\smash{\tailY{\pnt}{k}}} \leq \rr$.
\end{lemma}
\begin{proof}
    Let $\threshold = \rr/k^{1/\p}$, and let $y$ be the number of
    truncated coordinates in $\truncY{\pnt}{\threshold}$. If $y > k$
    then
    \begin{math}
        \norm{\p}{\truncY{\pnt}{\threshold}}%
        \geq%
        \sqrt[\p]{ y \threshold^\p }%
        > %
        \sqrt[\p]{ k \threshold^\p }%
        =%
        k^{1/\p} \threshold%
        =%
        \rr,
    \end{math}
    which is a contradiction. As such, all the non-zero coordinates of
    $\tailY{\pnt}{k}$ are present in $\truncY{\pnt}{\threshold}$, and
    we have that
    \begin{math}
        \norm{\p}{\smash {\tailY{\pnt}{k}} }%
        \leq%
        \norm{\p}{\truncY{\pnt}{\threshold}}%
        \leq%
        \rr.
    \end{math}
\end{proof}

\subsubsection{On the probability of a heavy point to be sampled as
   light, and vice versa}

\begin{lemma}
    \lemlab{basic:exp:var} %
    Let $\pnt = (\pntc_1,\ldots, \pntc_d)$ be a point in $\Re^d$, and
    consider a random $\subseq \in {\DistD{\pr}}$, see
    \defref{t:splay}.  We have that
    $\Ex{\bigl. \norm{\p}{ \subseq \pnt}^{\p}} = \pr
    \norm{\p}{\pnt}^{\p}$, and
    \begin{math}
        \Var{\bigl. \norm{\p}{ \subseq \pnt}^{\p}} = \pr(1-\pr)
        \norm{2\p}{\pnt}^{2\p}.
    \end{math}
\end{lemma}
\begin{proof}
    Let $X_i$ be a random variable that is $\cardin{\pntc_i}^{\p}$
    with probability $\pr$ and $0$ otherwise. For
    $Z = \norm{\p}{ \subseq \pnt}^{\p} $, we have that
    \begin{math}
        \Ex{ Z }%
        =%
        \sum_{i=1}^d \Ex{X_i}%
        =%
        \sum_{i=1}^d \pr \cardin{\pntc_i}^{\p}%
        =%
        \pr \norm{\p}{\pnt}^{\p}.
    \end{math}
%
    As for the variance, we have
    \begin{math}
        \Var{X_i}%
        =%
        \Ex{X_i^2} - \pth{\Ex{X_i}}^2 %
        = %
        \pr \cardin{\pntc_i}^{2\p} - \pr^2 \cardin{\pntc_i}^{2\p} =
        \pr(1-\pr) \cardin{\pntc_i}^{2\p}.
    \end{math}
    As such, we have
    \begin{math}
        \Var{Z}%
        =%
        \sum_{i=1}^d \Var{ X_i}%
        =%
        \pr(1-\pr) \sum_{i=1}^d %
        \cardin{\pntc_i}^{2\p}%
        =%
        \tTimes \pr(1-\pr) \norm{2\p}{\pnt}^{2\p}.
    \end{math} 
\end{proof}

\begin{lemma}%
    \lemlab{lp:t:heavy}%
    Let $\pnt$ be a point in $\Re^d$, and let $\threshold > 0$ be a
    number.  We have that
    $\norm{2\p}{\truncY{\pnt}{\threshold}}^{2\p} \leq \threshold^{\p}
    \norm{\p}{\pnt}^{\p}$.
\end{lemma}

\begin{proof}
    Consider the $\threshold$-truncated point
    $\pntA = \truncY{\pnt}{\threshold}$, see \defref{truncate}.  Each
    coordinate of $\pntA$ is smaller than $\threshold$, and thus
    \begin{math}
        \norm{2\p}{\pntA}^{2\p}%
        =%
        \sum_{i=1}^d \cardin{\pntAc_i}^{2\p}%
        \leq%
        \sum_{i=1}^d \threshold^{\p} \cardin{\pntAc_i}^{\p}%
        \leq%
        \threshold^{\p} \norm{\p}{\pntA}^{\p}.
    \end{math}
\end{proof}

\begin{lemma}
    \lemlab{lp:heavy:block}%
  Consider a sequence $\subseq \in \DistD{\pr}$.  If $\pnt$ is a
    $( \nfrac{\rr}{k^{1/\p}},\RR)$-heavy point and
    $\RR \geq (8\cCoord)^{1/\p} \rr$, then
    \begin{math}
        \Prob{\Bigl. \norm{\p}{ \subseq \pnt}^{\p} \geq {\tfrac{1}{2}}
           {\Ex{\bigl. \smash{ \norm{\p}{ \subseq
                       \truncY{\pnt}{\threshold}}^{\p}}} }}%
        \geq%
        {1}/{2},
    \end{math}
    where $\threshold = \rr/k^{1/\p}$.
\end{lemma}
\begin{proof}
    Consider the $\threshold$-truncated point
    $\pntA = \truncY{\pnt}{\threshold}$.  Since $\pnt$ is
    $(\nfrac{\rr}{ k^{1/\p}}, \RR)$-heavy, we have that
    $\norm{\p}{\pnt} \geq \norm{\p}{\pntA} \geq \RR$.  Now, setting
    $Z = \norm{\p}{ \subseq \pntA}^{\p}$, and using
    \lemref{basic:exp:var}, we have
    \begin{math}
        \mu = \Ex{ Z }%
        =%
        \pr \norm{\p}{\pntA}^{\p} %
    \end{math}
    and
    \begin{math}
        \sigma^2%
        = %
        \Var{Z } =%
        \pr(1-\pr) \norm{2\p}{\pntA}^{2\p}%
        \leq %
        \pr(1-\pr) \threshold^{\p} \norm{\p}{\pntA}^{\p},
    \end{math}
    by \lemref{lp:t:heavy}.  Now, we have that
    \begin{math}
        \Prob{ \bigl.\smash{\norm{\p}{ \subseq \pnt}^{\p} \leq
              \nfrac{\mu}{2}} }%
        \leq%
        \Prob{ \bigl. \smash{\norm{\p}{ \subseq \pntA}^{\p} \leq
              \nfrac{\mu}{2} } }.
    \end{math}
    As such, by Chebyshev's inequality, and since
    $\pr = 1/(\cCoord k)$, if $\RR \geq (8\cCoord)^{1/\p} \rr$, we
    have
    \begin{align*}
      \Prob{ Z \leq \frac{\mu}{2} }%
      \leq%
      \Prob{ \cardin{\bigl.Z - \mu} \geq
      \frac{\mu/2}{\sigma} \sigma }%
      \leq%
      \pth{\frac{\sigma}{\mu/2}}^2%
            \leq%
            \frac{{ \pr(1-\pr)  \threshold^\p  \norm{\p}{\pntA}^{\p}%
            }}{\pr^2 \norm{\p}{\pntA}^{2\p}/4 }%
            \leq%
            {\frac{ 4 \threshold^\p }{\pr \norm{\p}{\pntA}^{\p} }}%
            \leq%
            4 {\frac{ \cCoord k \rr^{\p} }{k \RR^{\p} }}%
            \leq%
            \frac{1}{2}.
    \end{align*}
\end{proof}

\begin{lemma}
    \lemlab{success}%
    Let $\pnt$ be a prespecified point. The probability that a
    sequence $\seq$ sampled from $\DistD{\pr}^\tTimes$ does not sample
    any of the $k$ heaviest coordinates of $\pnt$ is
    $\geq n^{- \cTimes/\cCoord - 2\cTimes/(\cCoord^2k) } \approx
    1/n^{\cTimes/\cCoord}$
    (for the simplicity of exposition, in the following, we use this
    rougher estimate).
\end{lemma}

\begin{proof}
    Let $S$ be the set of $k$ indices of the coordinates of $\pnt$
    that are largest (in absolute value). The probability that
    $\subseq_i$ does not contain any of these coordinates is
    $(1-\pr)^k$, and overall this probability is
    \begin{math}
        \val = \pth{1- \frac{1}{\cCoord k}}^{k \tTimes}%
        =%
        \pth{1- \frac{1}{\cCoord k}}^{ \cCoord k (\cTimes/\cCoord) \ln
           n}.%
    \end{math}
    Now, we have
    \begin{math}
        \Bigl.
        \val%
        \geq%
        \exp \pth{- \frac{\cTimes \ln n }{\cCoord} }%
        \pth{1- \frac{1}{\cCoord k}}^{ (\cTimes/\cCoord) \ln n}%
        \geq%
        n^{- \cTimes/\cCoord}%
        \exp \pth{- \frac{2\cTimes \ln n}{\cCoord^2 k}}%
        = %
        n^{- \cTimes/\cCoord - 2\cTimes/(\cCoord^2k) },
    \end{math}
    since, for any integer $m$, we have
    $(1-1/m)^{m-1} \geq 1/e \geq (1-1/m)^{m}$, and
    $e^{-2x} \leq 1-x \leq e^x$, for $x \in (0,1/2)$.
\end{proof}

\begin{lemma}
    \lemlab{lp:light:good}%
    Consider a point $\pnt$ such that
    $\norm{\p}{\smash{\tailY{\pnt}{k}}}\leq r$ (see
    \defref{tail}). Conditioned on the event of \lemref{success}, we
    have that
    $\Prob{ \Bigl. \norm{\p}{ \seq \pnt}^{\p} \geq 2 \tTimes \pr
       \rr^{\p} } \leq 1/2$, where $\seq \in \DistD{\pr}^\tTimes$.
\end{lemma}

\begin{proof}
    By \lemref{basic:exp:var} for
    $\Bigl.Z = \snorm{\p}{\seq \tailY{\pnt}{k}}^\p$, we have
    $\mu = \Ex{ Z } = \tTimes\pr \norm{\p}{\smash {\tailY{\pnt}{k}}
    }^\p \leq \tTimes \pr \rr^\p$.
    The desired probability is $\Prob{Z \geq 2 \mu} \leq 1/2$, which
    holds by Markov's inequality.
\end{proof}

\begin{lemma}
    \lemlab{lp:heavy:far}%
    Let $\RR \geq (8\cCoord)^{1/\p} \rr$. If $\pnt$ is a
    $( \nfrac{\rr}{k^{1/\p}},\RR)$-heavy point, then
    \begin{math}
        \Prob{\Bigl. \norm{\p}{ \seq \pnt}^{\p} \geq \nfrac{\tTimes
              \pr \RR^{\p}}{8} } \geq 1 - 2/n^{\cTimes/8}.
    \end{math}
\end{lemma}
\begin{proof}%
    Let $\threshold = \rr/k^{1/\p}$ and
    $\pntA = \truncY{\pnt}{\threshold}$, and for all $i$, let
    $Y_i = \norm{\p}{ \subseq_i \pntA}^{\p}$. By
    \lemref{lp:heavy:block}, with probability at least half, we have
    that
    \begin{math}
        Y_i %
        \geq \Ex{Y_i}/2%
        \geq%
        \pr \norm{\p}{\pntA}^{\p}/2%
        \geq%
        \pr \RR^{\p}/ 2.
    \end{math}
    In particular, let $Z_i = \min( Y_i, \pr \RR^{\p}/ 2)$, and
    observe that
    \begin{math}
        \Ex{Z_i}%
        \geq%
        (\pr \RR^{\p}/ 2) \Prob{Y_i \geq \pr \RR^{\p}/ 2}
        \geq%
        \pr \RR^{\p}/ 4.
    \end{math}
    Thus, we have that
    \begin{math}
        \mu = \Ex{Z}%
        =%
        \Ex{\sum_{i=1}^\tTimes Z_i }%
        \geq%
        \tTimes \pr \RR^{\p}/ 4.
    \end{math}
    Now set $\RRA = {\tTimes \pr \RR^{\p}/ 8}$ and note that
    $\mu \geq 2\RRA$.  Now, by Hoeffding's inequality, we have
    that
    \begin{align*}
      \Prob{\Bigl. \norm{\p}{ \seq \pnt}^{\p} \leq \RRA}%
      &\leq%
        \Prob{\Bigl. Z \leq \RRA}%
        \leq%
        \Prob{ \Bigl. \cardin{ Z - \mu} \geq \mu - \RRA}%
        \leq%
        2 \exp\pth{ - \frac{2(\mu - \RRA)^2 }{\tTimes \pth{ \pr
        \RR^{\p}/ 2}^2}}%
      \\&%
          \leq%
          2 \exp\pth{ - \frac{8( \mu/2)^2 }{\tTimes \pr^2 \RR^{2\p}}}%
          \leq%
          2 \exp\pth{ - \frac{8( \tTimes \pr \RR^{\p}/ 8)^{2} }{ \tTimes
          \pr^2 \RR^{2\p}}}%
          =%
          2 \exp\pth{ - \frac{ \cTimes \ln n }{8 }}%
          \leq%
          \frac{2}{n^{\cTimes/8}}.
    \end{align*}
\end{proof}

\subsubsection{Putting everything together}

\begin{lemma}
    \lemlab{lp:heavy:tail}
    Let $\cDSTimes \in (0,1)$ be a parameter.  One can build the
    data-structure described in \secref{sec:lp:alg} with the following
    guarantees.  For a query point $\query \in \Re^d$, let
    $\nnq \in \PntSet$ be its $k$-robust nearest neighbor in $\PntSet$
    under the $\Lp$ norm, and let
    $\rr = \snorm{\p}{\tailY{(\nnq-\query)}{k}}$. Then, with high
    probability, the query algorithm returns a point
    $\pntA \in \PntSet$, such that $\query - \pntA$ is a
    $\pth{\Bigl. \nfrac{\rr}{k^{1/\p}},\, O(\rr(\nnConst +
       \nfrac{1}{\cDSTimes})^{1/\p}) \bigr.}$-light.
    The data-structure performs $O\pth{n^{\cDSTimes} \log n}$ of
    $\nnConst^{1/\p}$-\ANN queries under $\Lp$-norm.
\end{lemma}

\begin{proof}
    We start with the painful tedium of binding the parameters.  For
    the bad probability, bounded by \lemref{lp:heavy:far}, to be
    smaller than $1/n$, we set $\cTimes = 16$.  For the good
    probability $1/n^{\cTimes/\cCoord}$ of \lemref{success} to be
    larger than $1/n^\cDSTimes$, implies
    $n^\cDSTimes \geq n^{\cTimes/\cCoord}$, thus requiring
    $\cCoord \geq \cTimes / \cDSTimes$. Namely, we set
    $\cCoord = \cTimes / \cDSTimes$.  Finally, \lemref{lp:heavy:far}
    requires
    \begin{math}
        \RR \geq (8\cCoord)^{1/\p} \rr = (128/\cDSTimes)^{1/\p}
        \rr.
    \end{math}
    Let $\lambda = \max ( \nfrac{128}{\cDSTimes},\, 16\nnConst)$ and
    let $\RR = \lambda^{1/\p} \rr$.
    
    For a query point $\query$, let $\nnq$ be its $k$-robust \NN, and
    let $S$ be the set of $k$ largest coordinates in
    $\pntB = \query - \nnq$. Let $\Event$ denote the event of sampling
    a projection $\seq_i \in \DistD{\pr}^\tTimes$ that does not
    contain any of the coordinates of $S$.  By \lemref{success}, with
    probability $p \approx 1/n^{\cTimes/\cCoord} = 1 /n^{\cDSTimes}$,
    the event $\Event$ happens for the data-structure $\DA_i$, for any
    $i$.

    As such, since the number of such data-structures built is   
    \begin{math}
        \DSTimes = \Theta( n^{\cDSTimes}\log n) = O\pth{\bigl. (\log
           n) / p },
    \end{math}
    we have that, by Chernoff inequality, with high probability, that
    there are at least $m = \Omega(\log n)$ such data structures, say
    $ \DA_1, \ldots, \DA_m$.

    Consider such a data-structure $\DA_i$.  The idea is now to ignore
    the coordinates of $S$ all together, and in particular, for a
    point $\pnt \in \Re^d$, let $\rmC{\pnt}{S} \in \Re^{d-k}$ be the
    point where the $k$ coordinates of $S$ are removed (as defined in
    \defref{r:m:c}).  Since by assumption
    $\snorm{\p}{\rmC{\pntB}{S}} =\snorm{\p}{\tailY{\pntB}{k}} \leq r$,
    by \lemref{lp:light:good}, with probability at least half, the
    distance of $\seq_i \nnq$ from $\seq_i \query$ is at most
    $\ell = (2 \tTimes \pr)^{1/\p} \rr$.  Since there are
    $\Omega( \log n)$ such data-structures, we know that, with high
    probability, in one of them, say $\DA_1$, this holds.  By
    \lemref{lp:heavy:far}, any point $\rmC{\pnt}{S}$ (of $\PntSet$),
    that is $(\nfrac{\rr}{k^{1/\p}}, \RR)$-heavy, would be in distance
    at least
    \begin{math}
        \ell' = (\tTimes \pr/8)^{1/\p} \RR \geq (2\nnConst
        \tTimes\pr)^{1/\p}\rr=%
        \nnConst^{1/\p} \cdot \ell
    \end{math}
    in the projection $\seq_1$ from the projected $\query$.  Since
    $\DA_1$ is a $\nnConst^{1/\p}$-\ANN data-structure under the
    $L_{\p}$ norm, we conclude that no such point can be returned,
    because the distance from $\query$ to $\nnq$ in this
    data-structure is smaller than $\nnConst^{1/\p} \ell \leq \ell'$.
    Note that since for the reported point $\pntA$, the point
    $\rmC{\pntA}{S}$ cannot be
    $(\nfrac{\rr}{k^{1/\p}},\lambda^{1/\p}\rr)$-heavy, and that the
    coordinates in $S$ can contribute at most
    $k\, (\nfrac{\rr}{k^{1/\p}})^{\p} = \rr^{\p}$. We conclude that
    the point $\pntA$ cannot be
    $(\nfrac{\rr}{k^{1/\p}},(\lambda+1)^{1/\p}\rr)$-heavy.  Thus, the
    data-structure returns the desired point with high probability.

    As for the query performance, the data-structure performs
    $\DSTimes$ queries in $\nnConst^{1/\p}$-\ANN data-structures.
\end{proof}

This lemma would translate to the following theorem using
\lemref{lp:light}.

\subsection{The result}

\begin{theorem}
    \thmlab{l:p}
    Let $\PntSet \subseteq \Re^d$ be a set of $n$ points with the
    underlying distance being the $\Lp$ metric, and $k >0 $,
    $\cDSTimes \in (0,1)$, and $\nnConst \geq 1$ be parameters.  One
    can build a data-structure for answering the $k$-robust \ANN
    queries on $\PntSet$, with the following guarantees:
    \begin{compactenum}[\,\,(A)]
        \item Preprocessing time/space is equal to the space/time
        needed to store $M = O\pth{n^{\cDSTimes} \log n}$
        data-structures for performing $\nnConst^{1/\p}$-\ANN queries
        under the $\Lp$ metric, for a set of $n$ points in
        $O( (d/ k) \log n)$ dimensions.

        \item The query time is dominated by the time it takes to
        perform $\nnConst^{1/\p}$-\ANN queries in the $M$ \ANN
        data-structures.

        \item For a query point $\query$, the data-structure returns,
        with high probability, a point $\pntA \in \PntSet$, such that
        if one ignores 
        \begin{math}
            O(k( \nfrac{1}{ \cDSTimes} + \nnConst))
        \end{math}
        coordinates, then the $\Lp$ distance between $\query$ and
        $\pntA$ is at most
        \begin{math}
            O(\rr (\nfrac{1}{\cDSTimes}+\nnConst)^{1/\p}),
        \end{math}
        where $\rr$ is the distance of the nearest neighbor to
        $\query$ when ignoring $k$ coordinates.  (Formally,
        $\query - \pntA$ is
        $\pth{\bigl.\nfrac{\rr}{k^{1/\p}}, O(\rr(\nnConst +
           \nfrac{1}{\cDSTimes})^{1/\p}) \bigr.}$-light.)
    \end{compactenum}
\end{theorem}

\begin{corollary}
    Setting $\nnConst=2$, the algorithm would report a point $\pntA$
    using $2^{1/\p}$-\ANN data-structures, such that if one ignores
    $O(k/\cDSTimes)$ coordinates, the $\Lp$ distance between $\query$
    and $\pntA$ is at most $O(\nfrac{r}{\cDSTimes^{1/\p}})$.
    Formally, $\query - \pntA$ is
    \begin{math}
        \pth{\bigl.\smash {\nfrac{\rr}{k^{1/\p}},\,
              O(\nfrac{r}{\cDSTimes^{1/\p}})}}
    \end{math}%
    -light.
\end{corollary}


\section{Budgeted version}
\seclab{sec:budgeted}
\subsection{Definition of the problem}

In this section, we consider the budgeted version of the problem for
$L_1$-norm. Here, a coordinate $i$ has a cost $\cc_i \geq 0$ of
ignoring it, and we have a budget of $1$, of picking the coordinates
to ignore (note that since we can safely remove all coordinates of
cost $\cc_{i}=0$, we can assume that $\cc_i >0$).  Formally, we have a
vector of \emphi{costs} $\vecC = \pth{\cc_1, \ldots, \cc_d}$, where
the $i$\th coordinate, $\cc_i$, is the cost of ignoring this
coordinate.  Intuitively, the cost of a coordinate shows how much we
are certain that the value of the coordinate is correct.

The set of \emphi{admissible projections}, is
\begin{align}
  \AdmS%
  =%
  \AdmS\pth{ \vecC }%
  =%
  \Set{ \CSet}{ \CSet \subseteq \IntRange{d} \text{ and }
  \Bigl. {\smash {\textstyle \sum^{}_{i \in \IntRange{d}
  \setminus \CSet}} }\, \cc_i \leq 1}.
  \eqlab{admissible}
\end{align}
Given two points $\pnt, \pntA$, their \emphi{admissible distance} is
\begin{align}
  \daZ{\vecC}{\pnt}{\pntA}%
  =%
  \min_{\CSet \in \AdmS} \norm{1}{\CSet (\pnt - \pntA)},
  \eqlab{a:dist}
\end{align}
where we interpret $\CSet$ as a projection (see \defref{set:proj}).

The problem is to find for a query point $\query$ and a set of points
$\PntSet$, both in $\Re^d$, the \emph{robust nearest-neighbor
   distance} to $\query$; that is,
\begin{align}
  \dnnZ{\vecC}{\query}{\PntSet}%
  =%
  \min_{\pnt \in \PntSet} \daZ{\vecC}{\pnt}{\pntA}.%
  \eqlab{n:n:n}
\end{align}
The point in $\PntSet$ realizing this distance is the \emphi{robust
   nearest-neighbor} to $\query$, denoted by
\begin{math}
    \nnq%
    =%
    \nnCX{\query}%
    =%
    \nnCZ{\vecC}{\query}{\PntSet}.
\end{math}
The unweighted version can be interpreted as solving the problem for
the case where all the coordinates have uniform cost $1/k$.

\begin{defn}
    \deflab{good:bad}%
    If $\nnq$ is the nearest-neighbor to $\query$ under the above
    measure, then the set of \emph{good} coordinates is
    \begin{math}
        \goodCoords(\query) = \arg \min_{\CSet \in \AdmS}
        \norm{1}{\CSet (\nnq - \query)} \subseteq \IntRange{d},
    \end{math}
    and the set of \emphi{bad} coordinates is
    \begin{math}
        \badCoords(\query) = \IntRange{d} \setminus \goodCoords(
        \query).
    \end{math}
\end{defn}


In what follows, we modify the algorithm for the unweighted case and
analyze its performance for the budgeted case. Interestingly, the
problem is significantly harder.

\subsubsection{Hardness and approximation of robust %
   distance for two points}

For two points, computing their distance is a special instance of
\ProblemC{Min-Knapsack}. The problem is \NPHard (which is well known),
as testified by the following lemma.%

\begin{lemma}
    \lemlab{2:points}%
    Given two points $\pnt, \pntA \in \Re^d$, and a cost vector
    $\vecC$, computing
    $\min_{\CSet \in \AdmS} \norm{1}{\CSet (\pnt - \query)}$ is
    \NPComplete, where $\AdmS$ is the set of admissible projections
    for $\vecC$ (see \Eqref{admissible}).
\end{lemma}

\begin{proof}
    This is well known, and we provide the proof for the sake of
    completeness.

    Consider an instance of \ProblemC{Partition} with integer numbers
    $b_1, \ldots, b_d$. Let $\alpha = \sum_{i=1}^d b_i$, and consider
    the point $\pnt = (2b_1/ \alpha, \ldots, 2b_d/\alpha)$, and set
    the cost vector to be $\vecC = \pnt$. Observe that
    $\norm{1}{\pnt} = \norm{1}{\vecC} = 2$. In particular, there is a
    point in robust distance at most $1$ from the origin, with the
    total cost of the omitted coordinates being $1$ $\iff$ the given
    instance of \ProblemC{Partition} has a solution.

    Indeed, consider the set of coordinates $\CSet$ realizing
    $\ell = \min_{\CSet \in \AdmS} \norm{1}{\CSet (\pnt - 0)}$.  Let
    $\badCoords = \IntRange{d} \setminus \CSet$, and observe that the
    cost of the omitted coordinates
    $\ell' = \norm{1}{\badCoords \vecC} = \norm{1}{\badCoords \pnt}$
    is at most $1$ (by the definition of the admissible set
    $\AdmS$). In particular, we have $\ell + \ell' = \norm{1}{\pnt}=2$
    and $\ell' \leq 1$. As such, the minimum possible value of $\ell$
    is $1$, and if it is $1$, then $\ell = \ell' = 1$, and $\CSet$ and
    $\badCoords$ realize the desired partition.
\end{proof}

Adapting the standard \PTAS for subset-sum for this problem, readily
gives the following.

\begin{lemma}[\cite{i-aamkp-09}]
    \lemlab{i-aamkp}%
    Given points $\pnt, \pntA \in \Re^d$, and a cost vector
    $\vecC \in [0,1]^d$, one can compute a set
    $\CSet \in \AdmS \pth{\vecC}$, such that
    \begin{math}
        \daZ{\vecC}{\pnt}{\pntA}%
        \leq%
        \norm{1}{\CSet( \pnt - \pntA)}%
        \leq%
        (1+\eps) \daZ{\vecC}{\pnt}{\pntA}.
    \end{math}
    The running time of this algorithm is $O(d^{\,4}/\eps)$.
\end{lemma}

\subsubsection{Embedding with scaling}

Given a vector $\vecC \in [0,1]^d$ , consider generating a sequence
$\subseq$ of integers $i_1 < \cdots < i_k$, by picking the number
$i \in \IntRange{d}$, into the sequence, with probability $\cc_i$.  We
interpret this sequence, as in \defref{set:proj}, as a projection,
except that we further scale the $j$\th coordinate, by a factor of
$1/\vecC_{i_j}$, for $j=1,\ldots, k$. Namely, we project the $i$\th
coordinate with probability $\vecC_i$, and if so, we scale it up by a
factor of $1/\vecC_{i}$, for $i=1,\ldots, d$ (naturally, coordinates
with $\cc_i=0$ would never be picked, and thus would never be
scaled). Let $\DistB{\vecC}$ denote this distribution of weighted
sequences (maybe a more natural interpolation is that this is a
distribution of projection matrices).

\begin{observation}
    \obslab{norm:1}%
    Let $\vecC \in (0,1]^d$ be a cost vector with non zero entries.
    For any point $\pnt \in \Re^d$, and a random
    $\subseq \in \DistB{\vecC}$, we have that
    \begin{math}
        \Ex{ \bigl. \norm{1}{\subseq \pnt } } = \norm{1}{\pnt}.
    \end{math}
\end{observation}

\subsection{Algorithm}

The input is a point set $\PntSet$ of $n$ points in $\Re^d$.  Let
$\vecC = \pth{\cc_1,\cdots,\cc_d}$ be the vector of costs, and
$\cCoord, \cTimes, \cDSTimes$ be three constants to be specified
shortly, such that $\cDSTimes \in (0,1)$.  Let
$\tTimes = \cTimes \ln n$, and $\DSTimes = O( n^{\cDSTimes} \log n)$.

\paragraph{Preprocessing.}

We use the same algorithm as before.  We sample $\DSTimes$ sequences
$\seq_1, \ldots, \seq_\DSTimes \in \DistBY{\cCoord}{\tTimes}$
Then we embed the point set using these projections, setting
$\PntSet_i = \seq_i \PntSet$, for $i=1,\ldots, \DSTimes$. Next, we
preprocess the point set $\PntSet_i$ for $2$-\ANN queries under the
$L_1$-norm, and let $\DA_i$ be the resulting \ANN data-structure, for
$i=1, \ldots, \DSTimes$.

\paragraph{Answering a query.} %

Given a query point $\query$, the algorithm performs a $2$-\ANN query
for $\seq_i \query$ in $\DA_i$, this \ANN corresponds to some original
point $\pnt_i \in \PntSet$, for $i=1,\ldots, \DSTimes$. The algorithm
then $(1+\eps)$-approximate the distance
$\ell_i = \daZ{\vecC}{\query}{\pnt_i}$, for $i=1,\ldots, \DSTimes$,
using the algorithm of \lemref{i-aamkp}, and returns the point
realizing the minimum distance as the desired \ANN. 

\subsection{The analysis}

\begin{lemma}
    \lemlab{budget:success} For a query point $\query$, let
    \begin{math}
        \badCoords= \badCoords(\query)
    \end{math}
    be the bad coordinates of $\query$ (thus
    $\norm{1}{\badCoords \vecC} \leq 1$).  Then the probability that
    $\seq \in \DistBY{\cCoord}{\tTimes}$ misses all the coordinates of
    $\badCoords$ is at least $n^{-2\cTimes/\cCoord}$.
\end{lemma}

\begin{proof}
    Let $j_1, \ldots, j_w$ be the bad coordinates in $\badCoords$. We
    have that $\sum_{i=1}^w \cc_{j_i} \leq 1$. As such, the
    probability that $\subseq \in \DistBY{\cCoord}{}$ fails to sample
    a coordinate in $\badCoords$ is
    \begin{math}
        \prod_{i=1}^{w} \pth{ 1-\cc_{j_i}/\cCoord}%
        \geq%
        \prod_{i=1}^{w} \exp\pth{ -2 \cc_{j_i}/\cCoord } \geq%
        \exp\pth{ \sum_{i=1}^w -2 \cc_{j_i}/\cCoord } \geq%
        \exp\pth{ -2 /\cCoord },
    \end{math}
    as $\cc_{j_i}/\cCoord \leq 1/2$, and $1-x \geq \exp(-2x)$, for
    $x \in (0,1/2)$.  As such, the probability that a sequence
    $\seq = \pth{\subseq_1, \ldots, \subseq_\tTimes} \in
    \DistBY{\cCoord}{\tTimes}$
    avoids $\badCoords$ is at least
    $\exp\pth{ -2\tTimes /\cCoord } = \exp\pth{ - 2 (\cTimes/\cCoord)
       \ln n } = n^{-2\cTimes/\cCoord}$.
\end{proof}

Let $\goodCoords = \goodCoords(\query)$ and let
\begin{math}
    \rr%
    =%
    \norm{1}{\goodCoords\pth{\query - \nnq}}%
\end{math}
be the distance from $\query$ to its nearest neighbor only considering
the coordinates in $\goodCoords$.

\begin{lemma}
    \lemlab{e:v:w}%
    For a point $\pnt \in \Re^d$, and a random
    $\seq \in \DistBY{\cCoord}{\tTimes}$, we have that
    \begin{math}
        \Ex{\norm{1}{\seq \pnt}} = \tTimes \norm{1}{\pnt},
    \end{math}
    and
    \begin{math}
        \Var{\bigl. \norm{1}{\seq \pnt}}%
        =%
        \tTimes \sum_{i=1}^d \pnt_i^2 \pth{ \cCoord/\cc_i -1},
    \end{math}
\end{lemma}
\begin{proof}
    The claim on the expectation follows readily from
    \obsref{norm:1}. As for the variance, let $i$ be a coordinate that
    has non-zero cost, and let $X_i$ be a random variable that is
    $\cCoord/\cc_i$ with probability $\cc_i/\cCoord$, and $0$
    otherwise. We have that
    \begin{math}
        \Var{X_i}%
        =%
        \Ex{X_i^2} - \pth{\Ex{X_i}}^2 %
        =%
        \pth{\Bigl. \frac{\cc_i}{\cCoord} (\cCoord/\cc_i)^2 +
           (1-\cc_i/\cCoord) 0} - 1^2%
        =%
        \cCoord/\cc_i -1.
    \end{math}
    As such for $\subseq \in \DistBY{\cCoord}{\tTimes}$, we have that
    \begin{math}
        \Var{\Bigl.\norm{1}{\subseq \pnt}}%
        =%
        \Var{ \sum_{i=1}^d \pnt_i X_i}%
        =%
        \sum_{i=1}^d \Var{ \pnt_i X_i}%
        =%
        \sum_{i=1}^d \pnt_i^2 \Var{ X_i}%
        =%
        \sum_{i=1}^d \pnt_i^2 \pth{ \cCoord/\cc_i -1}
    \end{math}
    and thus the lemma follows.
\end{proof}

As before, we want to avoid giving too much weight to a single
coordinate which might have a huge (noisy) value in it. As such, we
truncate coordinates that are too large. Here, things become somewhat
more subtle, as we have to take into account the probability of a
coordinate to be picked.

\begin{defn}
    For a cost vector $\vecC \in \Re^d$, a positive number $\rr > 0$,
    and a point $\pnt \in \Re^d$, let
    \begin{math}
        \trcZ{\rr}{\vecC}{\pnt} %
        =%
        \pth{\pntc_1',\ldots, \pntc_d'}
    \end{math}
    be the \emphi{truncated} point, where
    \begin{align}
      \ds%
      \pntc_i' = \min\pth{ \cardin{\pntc_i}, \frac{\rr}{\cCoord
      /\cc_i -1}}, \eqlab{t:w}%
    \end{align}
    for $i=1,\ldots, d$.
\end{defn}

The truncation seems a bit strange on the first look, but note that a
coordinate $i$ that has a cost $\cc_i$ approaching $0$, is going to be
truncated to zero by the above. Furthermore, it ensures that a
``heavy'' point would have a relatively small variance in the norm
under the projections we use, as testified by the following easy
lemma.

\begin{lemma}
    \lemlab{stupid}%
    Consider a point $\pnt \in \Re^d$, and a random
    $\subseq \in \DistB{\vecC/\cCoord}$. For,
    $\pnt' = \trcZ{\rr}{ \vecC}{\pnt} =( \pntc_1', \ldots, \pntc_d')$,
    consider the random variable $X = \norm{1}{\subseq \pnt'}$.  We
    have that
    \begin{math}
        \Ex{\bigl.X}%
        =%
        \norm{1}{\pnt'},
    \end{math}
    \begin{math}
        \Bigl.\Var{\bigl. X}%
        \leq%
        \rr \norm{1}{\pnt'},
    \end{math}
    and
    \begin{math}
        \Prob{\bigl. X \geq \Ex{X}/2 } \geq 1 - 4 r / \norm{1}{\pnt'}.
    \end{math}
\end{lemma}
\begin{proof}
    The bound on the expectation is by definition. As for the
    variance, by \lemref{e:v:w} and \Eqref{t:w}, we have
    \begin{math}
        \ds%
        \Var{\Bigl. \norm{1}{\subseq \pnt'}}%
        =%
        \sum_{i=1}^d \pth{\pntc_i'}^2 \pth{ \cCoord/\cc_i -1}%
        \leq%
        \sum_{i=1}^d \pntc_i' \frac{\rr}{\cCoord /\cc_i -1} \pth{
           \cCoord/\cc_i -1}%
        =%
        \rr \norm{1}{\pnt'}.
    \end{math}
    
    Let $\mu = \Ex{X }$, and $\sigma^2 = \Var{ X }$. By Chebyshev's
    inequality, we have that
    \begin{align*}
      \Prob{ X \geq \frac{\mu}{2} }%
      &\geq%
        1 - \Prob{ \cardin{\bigl.\norm{1}{ \subseq \pnt'} - \mu} \geq
        \frac{\mu/2}{\sigma} \sigma }%
        \geq%
        1 - \pth{\frac{\sigma}{\mu/2}}^2%
        =%
        1 - 4 \frac{\Var{X}}{\mu^2 } .%
    \end{align*}
    By the above, we have that
    \begin{math}
        {\Var{X}}/{\mu^2 }%
        \leq%
        \rr \norm{1}{\pnt'} / \norm{1}{\pnt'}^2.
    \end{math}
\end{proof}

\begin{defn}
    \deflab{light:heavy} %
    For a value $\val > 0$, a point $\pnt$ is
    \emphi{$(\rr, \vecC, \val)$-light} (resp.~\emphi{heavy}) if
    $\norm{1}{\pnt'} \leq \val$ (resp.~$\norm{1}{\pnt'} \geq \val$),
    where $\pnt' = \trc^{}_{\rr, \vecC} (\pnt)$.
\end{defn}

\begin{lemma}
    \lemlab{stupid2}%
    Consider a point $\pnt \in \Re^d$ that is
    $(\rr, \vecC, \RR)$-heavy, for $\RR \geq8\rr$, and a
    random $\seq \in \DistBY{\cCoord}{\tTimes}$. Then, we have that
    \begin{math}
        \Prob{ \bigl. \norm{1}{\seq \pnt} \geq {\tTimes \RR}/{8}}
        \geq%
        1 - {1}/{n^{\beta/2}}.
    \end{math}
\end{lemma}

\begin{proof}
    Let $\seq = \pth{\subseq_1, \ldots, \subseq_\tTimes}$.  Let
    $X_i = \norm{1}{\subseq_i \pnt' }$, and let $Y_i$ be an indicator
    variable that is one if $X_i \geq \RR/2$. We have, by
    \lemref{stupid}, that
    \begin{align*}
      \Prob{\Bigl. Y_i=1} %
      &=%
        \Prob{\bigl. X_i \geq \frac{\RR}{2} }%
        \geq %
        \Prob{\bigl.  \norm{1}{\subseq_i \pnt' } \geq \frac{\RR}{2}
        }%
        \geq%
        \Prob{\bigl.  \norm{1}{\subseq_i \pnt' } \geq
        \frac{\Ex{\norm{1}{\subseq_i \pnt' }}}{2} }%
        \geq%
        1 - \frac{4 \rr}{ \norm{1}{\pnt'}}%
      \\&%
          =%
          1 - \frac{4 \rr}{ \RR}%
          \geq%
          \frac{1}{2},
    \end{align*}
    as $\RR \geq 8\rr$.  By Chernoff inequality, we have that
    \begin{math}
        \Prob{\sum_{i=1}^\tTimes Y_i \leq \tTimes /4}%
        \leq%
        \exp \pth{ - 2 (\tTimes/2)^2 / \tTimes }%
        =%
        \exp \pth{ - \tTimes/2}%
        =%
        {1}/{n^{\beta/2}}
    \end{math}
    since $\tTimes = \cTimes \ln n$.

    In particular, we have that
    \begin{math}
        \norm{1}{\seq \pnt} %
        \geq%
        \sum_{i=1}^\tTimes Y_i \frac{\RR}{2}, %
    \end{math}
    and as such
    \begin{align*}
      \Prob{ \bigl. \norm{1}{\seq \pnt} \geq {\tTimes \RR}/{8}}
      \geq%
      \Prob{\sum\nolimits_{i=1}^\tTimes Y_i \geq \tTimes /4}%
      \geq%
      1 - {1}/{n^{\beta/2}}.
    \end{align*}
\end{proof}

\subsection{The result}

\begin{theorem}
    \thmlab{main:budgeted} %
    Let $\PntSet$ be a point set in $\Re^d$, let $\vecC \in [0,1]^d$
    be the cost of the coordinates, and let $\eps >0$ be a
    parameter. One can build a data-structure, such that given a query
    point $\query$, it can report a robust approximate
    nearest-neighbor under the costs of $\vecC$. Formally, if
    $\nnq = \nnCX{\query}$ is the robust nearest-neighbor (see
    \Eqrefpage{n:n:n}) when one is allowed to drop coordinates of
    total cost $1$, and its distance to this point is
    $\rr = \daZ{\vecC}{\query}{\nnq}$ (see \Eqref{a:dist}), then the
    data-structure returns a point $\pnt$, such that $\query - \pnt$
    is $(\rr, \vecC, 33(1+\eps)\rr)$-light (see
    \defref{light:heavy}). The data-structure has the following
    guarantees:
    \begin{compactenum}[\qquad(A)]
        \item The preprocessing time and space is
        $\tldO(n^{\delta}) T_2(n)$, where $T_2(n)$ is the
        preprocessing time and space needed to build a single
        data-structure for answering (standard) $2$-\ANN queries in
        the $L_1$-norm for $n$ points in $d$ dimensions.

        \item The query time is
        $O(n^{\delta} Q_2 + n^{\delta} d^4 /\eps)$, where $Q_2$ is the
        query time of answering $2$-\ANN queries in the above \ANN
        data-structures.
    \end{compactenum}
\end{theorem}

\begin{proof}
    The proof is similar to \lemref{lp:heavy:tail} in the unweighted
    case. We set $\cTimes = 4$, $\cCoord = 2\cTimes/\cDSTimes$, and
    $R\geq 32r$. By the same arguments as the unweighted case, and
    using \lemref{budget:success}, \lemref{e:v:w}, and Markov's
    inequality, with high probability, there exists a data-structure
    $\DA_i$ that does not sample any of the bad coordinates
    $\badCoords(\query)$, and that
    $\norm{1}{\seq_i (\query-\nnq)}\leq 2tr$. By \lemref{stupid2} and
    union bound, for all the points $\pnt$ such that
    $\rmC{\pnt}{\badCoords(\query)}$ (see \defref{r:m:c}) is
    $(\rr, \vecC, 32\rr)$-heavy, we have
    $\norm{1}{\seq_i \pnt}\geq 4tr$. Thus by \lemref{i-aamkp} no such
    point would be retrieved by $\DA_i$. Note that since for the
    reported point $\pnt$, we have that
    $\rmC{\pnt}{\badCoords(\query)}$ is $(\rr, \vecC, 32\rr)$-light,
    and that
    \begin{align*}
      \sum_{i\in \badCoords(\query)} \frac{\rr}{\cCoord/\cc_i - 1}
      \leq%
      \rr \sum_{i\in \badCoords(\query)} \frac{\cc_i}{\cCoord -
      \cc_i}%
      \leq%
      \frac{\rr}{\cCoord-1} \sum_{i\in \badCoords(\query)} \cc_i
      \leq%
      \frac{\rr}{\cCoord-1} \leq \rr,
    \end{align*}
    the point $\pnt$ is $(\rr, \vecC, 33\rr)$-light.  Using
    \lemref{i-aamkp} implies an additional $(1+\eps)$ blowup in the
    computed distance, implying the result.
\end{proof}

\begin{corollary}
    Under the conditions and notations of \thmref{main:budgeted}, for
    the query point $\query$ and its returned point $\pnt$, there
    exists a set of coordinates $\CSet \subseteq \IntRange{d}$ of cost
    at most $O(1)$ (i.e.,    
    \begin{math}
        \Bigl. {\smash {\textstyle \sum^{}_{i \in\CSet}} }\, \cc_i
        =O(1),
    \end{math}
    such that $\snorm{1}{\rmC{(\query - \pnt)}{\CSet}}\leq
    O(\rr)$.
    That is, we can remove a set of coordinates of cost at most $O(1)$
    such that the distance of the reported point $\pnt$ from the query
    $\query$ is at most $O(\rr)$.
\end{corollary}
\begin{proof}
    Let $\pntA = \query - \pnt$ and by \thmref{main:budgeted} $\pntA$
    is $(\rr, \vecC, c'\rr))$-light for some constant $c'$. Let
    $\pntA' = \trc^{}_{\rr, \vecC} (\pntA)$ and let $\CSet$ be the set
    of coordinates being truncated (i.e., all $i$ such that
    $\cardin{\pntAc_i}\geq \frac{\rr}{\cCoord/\cc_i -1}$). Clearly,
    the weight of the coordinates not being truncated is at most
    $\Bigl. {\smash {\textstyle \sum^{}_{i \in\IntRange{d}\setminus
             \CSet}} }\, \cardin{\pntAc_i} \leq \norm{1}{\pntA'} \leq
    c'\rr$.
    Also for the coordinates in the set $\CSet$, we have that
    $\Bigl. {\smash {\textstyle \sum^{}_{i \in\CSet}} }\,
    \frac{\rr}{\cCoord/\cc_i -1} \leq c'\rr$.  Therefore,
    \begin{math}
        \Bigl. {\smash {\textstyle \sum^{}_{i \in\CSet}} }\, \cc_i
        \leq \Bigl. {\smash {\textstyle \sum^{}_{i \in\CSet}} }\,
        \frac{\cCoord}{\cCoord -\cc_i}\cdot \cc_i%
        \leq%
        \cCoord c' := c,
    \end{math}
    assuming that $\cCoord \geq 2$, and noting that
    $0\leq \cc_i\leq 1$.
\end{proof}



\section{Application to data sensitive L{S}H queries}
\seclab{d:s:l}

Given a set of $n$ points $\PntSet$ and a radius parameter $\rr$, in
the approximate near neighbor problem, one has to build a
data-structure, such that for any given query point $\query$, if there
exists a point in $\PntSet$ that is within distance $\rr$ of $\query$,
it reports a point from $\PntSet$ which is within distance
$\rr(1+\eps)$ of the query point. In what follows we present a data
structure based on our sampling technique whose performance depends on
the relative distances of the query from all the points in the
data-set.

\subsection{Data structure}

\paragraph{Input.}

The input is a set $\PntSet$ of $n$ points in the hamming space
$\Hd = \brc{0,1}^d$, a radius parameter $\rr$, and an approximation
parameter $\eps > 0$.

\begin{remark:unnumbered}
    In the spirit of \remref{compress}, one can generate the
    projections in this case directly.  Specifically, for any value of
    $i$, consider a random projection
    \begin{math}
        \seq = (\subseq_1, \allowbreak \ldots, \subseq_i) \in
        \DistD{1/\rr}^i.
    \end{math}
    Two points $\pnt, \pntA \in \Hd$ \emphi{collide} under $\seq$, if
    $\subseq_j \pnt = \subseq_j \pntA$, for $j=1,\ldots, i$. Since we
    care only about collisions (and not distances) for the projected
    points, we only care what subset of the coordinates are being
    copied by this projection. That is, we can interpret this
    projection as being the projection $\cup_j \subseq_j$, which can
    be sampled directly from $\DistD{\pr}$, where $\pr = 1-(1-1/r)^i$.
    As such, computing $\seq$ and storing it takes $O(d)$
    time. Furthermore, for a point $\pnt \in \Hd$, computing
    $\seq \pnt$ takes $O(d)$ time, for any projection
    $\seq \in \DistD{1/\rr}^i$.
\end{remark:unnumbered}

\paragraph{Preprocessing.}%

For $i\geq 1$, let
\begin{align}
  \cpi{i} ( \ell) = (1-1/\rr)^{\ell i} \leq \exp \pth{ - \ell i
  /\rr}%
  \qquad%
  \text{ and }%
  \qquad%
  \Ni{i} =  \frac{\constA \log n} {\cpi{i}(\rr)}
  =%
  O\pth{ e^i \log n},%
  \eqlab{f:N}%
\end{align}
where $\constA$ is a sufficiently large constant.  Here, $\cpi{i}$ is
the \emph{collision probability function} of two points at distance
$\ell$ under projection $\subseq \in \DistD{1/\rr}^i$, and $\Ni{i}$ is
the number times one has to repeat an experiment with success
probability $\cpi{i}(\rr)$ till it succeeds with high probability.
Let $\N=\frac{\ln n}{1+\eps}$.

For $i=1,\ldots, \N$, compute a set
\begin{math}
    \Family_i = \brc{\bigl.\subseq^{}_1,\subseq^{}_2, \ldots,
       \subseq^{}_{\Ni{i}} \in \DistD{1/\rr}^i }
\end{math}
of projections.  For each projection $\subseq \in \Family_i$, we
compute the set $\subseq \PntSet$ and store it in a hash table
dedicated to the projection $\subseq$. Thus, given a query point
$\query \in \Hd$, the set of points \emph{colliding} with $\query$ is
the set
\begin{math}
    \colY{\query}{\subseq} = \Set{ \pnt \in \PntSet}{ \subseq \query =
       \subseq \pnt},
\end{math}
stored as a linked list, with a single entry in the hash table of
$\subseq$.  Given $\query$, one can extract, using the hash table, in
$O( d)$ time, the list representing $\colY{\query}{\subseq}$.  More
importantly, in $O(d)$ time, one can retrieve the size of this list;
that is, the number $\cardin{\colY{\query}{\subseq}}$.  For any
$i\leq \N$, let $\DS_i$ denote the constructed data-structure.

\paragraph{Query.}

Given a query point $\query \in \Hd$, the algorithm starts with $i=1$,
and computes, the number of points colliding with it in
$\Family_i$. Formally, this is the number
\begin{math}
    X_i = \sum_{\subseq \in \Family_i}
    \cardin{\colY{\query}{\subseq}}\!.
\end{math}
If $X_i > \constB\Ni{i}$, the algorithm increases $i$, and continues
to the next iteration, where $\constB$ is any constant strictly larger
than $e$.

Otherwise, $X_i \leq \constB \Ni{i}$ and the algorithm extracts from
the $\Ni{i}$ hash tables (for the projections of $\Family_i$) the
lists of these $X_i$ points, scans them, and returns the closest point
encountered in these lists.

The only remaining situation is when the algorithm had reached the
last data-structure for $\Family_\N$ without success. The algorithm
then extracts the collision lists as before, and it scans the lists,
stopping as soon as a point of distance $\leq (1+\eps)\rr$ had been
encountered.  In this case, the scanning has to be somewhat more
careful -- the algorithm breaks the set of projections of $\Family_\N$
into $T = \constA \log n$ blocks $B_1, \ldots, B_T$, each containing
$1/\cpi{\N}(\rr)$ projections, see \Eqref{f:N}. The algorithm computes
the total size of the collision lists for each block, separately, and
sort the blocks in increasing order by the number of their
collisions. The algorithm now scans the collision lists of the blocks
in this order, with the same stopping condition as above.

\begin{remark}
    There are various modifications one can do to the above algorithm
    to improve its performance in practice. For example, when the
    algorithm retrieves the length of a collision list, it can also
    retrieve some random element in this list, and compute its
    distance to $\query$, and if this distance is smaller than
    $(1+\eps)\rr$, the algorithm can terminate and return this point
    as the desired approximate near-neighbor. However, the advantage
    of the variant presented above, is that there are many scenarios
    where it would return the \emph{exact} nearest-neighbor. See below
    for details.
\end{remark}

\subsection{Analysis}

\subsubsection{Intuition}

The expected number of collisions with the query point $\query$, for a single
$\subseq \in \Family_i$, is
\begin{align}
  \sCi{i} = \Ex{ \Bigl. \cardin{\colY{\query}{\subseq}  }}%
  =%
  \sum_{\pnt \in \PntSet} \cpi{i}\pth{\bigl. \distH{\query}{\pnt}}
  \leq%
  \sum_{\pnt \in \PntSet} \exp \pth{ \bigl. - \distH{\query}{\pnt} i /\rr}.
  \eqlab{coll}
\end{align}
This quantity can be interpreted as a convolution over the point set.
Observe that as $\cpi{i}(\ell)$ is a monotonically decreasing function
of $i$ (for a fixed $\ell$), we have that
$\sCi{1} \geq \sCi{2} \geq \cdots$.

The expected number of collisions with $\query$, for all the
projections of $\Family_i$, is
\begin{align}
  \nCollX{i}%
  =%
  \Ex{X_i}%
  =%
  \Ni{i}  \sCi{i}.
  \eqlab{collisions}%
\end{align}
If we were to be naive, and just scan the lists in the $i$\th level,
the query time would be $O\pth{\bigl. d (X_i + \Ni{i})}$.  As such, if
$X_i = O(\Ni{i})$, then we are ``happy'' since the query time is
small. Of course, a priori it is not clear whether $X_i$ (or, more
specifically, $\nCollX{i}$) is small.

Intuitively, the higher the value $i$ is, the stronger the
data-structure ``pushes'' points away from the query point. If we are
lucky, and the nearest neighbor point is close, and the other points
are far, then we would need to push very little, to get $X_i$ which is
relatively small, and get a fast query time. The standard \LSH
analysis works according to the worst case scenario, where one ends up
in the last layer $\Family_\N$.

\begin{example}[If the data is locally low dimensional]
    \exmlab{low:dim}%
    The quantity $\nCollX{i}$ depends on how the data looks like near
    the query. For example, assume that locally near the query point,
    the point set looks like a uniform low dimensional point
    set. Specifically, assume that the number of points in distance
    $\ell$ from the query is bounded by
    $\#(\ell) = O\pth{ (\ell/\rr)^{k} }$, where $k$ is some small
    constant and $\rr$ is the distance of the nearest-neighbor to
    $\query$. We then have that
    \begin{align*}
      \nCollX{i}%
      &=%
        \Ni{i}   \sum_{\pnt \in \PntSet} \cpi{i}\pth{\bigl. \distH{\query}{\pnt}} %
        \leq%
        O(e^i \log n) \sum_{j=1}^{\infty} \#(j \rr) 
        \exp \pth{ - j i }%
        =%
        \sum_{j=1}^{\infty}  O\pth{ j^{k} \exp \pth{ i - j i  } \ln n}.%
    \end{align*}
    By setting $i = k$, we have
    \begin{align*}
      \nCollX{k}%
      \leq %
      \sum_{j=1}^{\infty}  O\pth{ \pth{ \frac{j}{e^{j-1}}}^{k} \ln n}
      \leq%
      \sum_{=1}^{\infty}  O\pth{ \frac{j}{2^{j}} \ln n}%
      \leq%
      O\pth{\ln n}
    \end{align*}
    Therefore, the algorithm would stop in expectation after $O(k)$
    rounds. 

    Namely, if the data near the query point locally behaves like a
    low dimensional uniform set, then the expected query time is going
    to be $O( d \log n)$, where the constant depends on the
    data-dimension $k$.
\end{example}

\subsubsection{Analysis}
\begin{lemma}
    \lemlab{fast}%
    If there exists a point within distance $\rr$ of the query point,
    then the algorithm would compute, with high probability, a point
    $\pnt$ which is within distance $(1+\eps)\rr$ of the query point.
\end{lemma}
\begin{proof}
    Let $\nnq \in \PntSet$ be the nearest neighbor to the query
    $\query$.  For any data-structure $\DS_i$, the probability that
    $\nnq$ does not collide with $\query$ is at most
    \begin{math}
        \pth{1-\cpi{i}(\rr)}^{\Ni{i}}%
        \leq%
        \exp\pth{ - \cpi{i}(\rr) \Ni{i} }%
        =%
        \exp\pth{ - O( \log n)}%
        \leq%
        1/n^{O(1)}.
    \end{math}
    Since the algorithm ultimately stops in one of these
    data-structures, and scans all the points colliding with the query
    point, this implies that the algorithm, with high probability,
    returns a point that is in distance $\leq (1+\eps) \rr$.
\end{proof}

\begin{remark}
    An interesting consequence of \lemref{fast} is that if the
    data-structure stops before it arrives to $\DS_{\N}$, then it
    returns the \emph{exact} nearest-neighbor -- since the
    data-structure accepts approximation only in the last
    level. Stating it somewhat differently, only if the data-structure
    gets overwhelmed with collisions it returns an approximate answer.
\end{remark}

\begin{remark}
    \remlab{m:large}%
    One can duplicate the coordinates $O(c)$ times, such that the
    original distance $\rr$ becomes $c \rr$.  In particular, this can
    be simulated on the data-structure directly without effecting the
    performance. As such, in the following, it is safe to assume that
    $\rr$ is a sufficiently large -- say larger than $\log n\geq\N$.
\end{remark}

\begin{lemma}
    \lemlab{N:i}%
    For any $i> 0$, we have $\Ni{i} = O\pth{e^i \log n}$.
\end{lemma}
\begin{proof}
    We have that
    \begin{math}
        \cpi{i}( \rr)%
        =%
        (1-1/\rr)^{(\rr-1)\rr i/(\rr-1)}%
        \geq%
        \exp\pth{ - \rr i/(\rr-1)}%
        =%
        \exp\pth{ - i - i/(\rr-1)} %
    \end{math}
    since for any positive integer $m$, we have
    $(1-1/m)^{m-1} \geq 1/e \geq (1-1/m)^{m}$.  As such, since we can
    assume that $\rr > \log n \geq i$, we have that
    \begin{math}
        \cpi{i}( \rr) = \Theta( \exp(-i) ).
    \end{math}
    Now, we have
    \begin{math}
        \Ni{i} = O( \frac{\log n} {\cpi{i}(\rr)})%
        =%
        O \pth{ e^i \log n }.
    \end{math}
\end{proof}

\begin{lemma}
    \lemlab{worst:case}%
    For a query point $\query$, the worst case query time is
    $O( dn^{1/(1+\eps)} \log n )$, with high probability.
\end{lemma}
\begin{proof}
    The worst query time is realized when the data-structure scans the
    points colliding under the functions of $\Family_\N$.
    
    We partition $\PntSet$ into two point sets:
    \smallskip%
    \begin{compactenum}[\qquad (i)]
        \item The close points are
        $\Pclose = \Set{\pnt \in \PntSet}{\distH{\query}{\pnt} \leq
           (1+\eps)\rr}$, and
        \item the far points are
        $\Pfar = \Set{\pnt \in \PntSet}{\distH{\query}{\pnt} >
           (1+\eps)\rr}$.
    \end{compactenum}
    \smallskip%
    Any collision with a point of $\Pclose$ during the scan terminates
    the algorithm execution, and is thus a \emph{good} collision. A
    \emph{bad} collision is when the colliding point belongs to
    $\Pfar$.

    Let $B_1, \ldots, B_{T}$ be the partition of the projections of
    $\Family_\N$ into blocks used by the algorithm.  For any $j$, we
    have
    \begin{math}
        \cardin{B_j} = O(\Ni{\N}/\log n) = O\pth{e^\N } =
        O\pth{n^{1/(1+\eps)}},
    \end{math}
    since $\N = (\ln n )/(1+\eps)$ and by \lemref{N:i}. Such a block,
    has probability of
    $p = (1-\cpi{\N}(\rr))^{1/\cpi{\N}(\rr)} \leq 1/e$ to not have the
    nearest-neighbor to $\query$ (i.e., $\nnq$) in its collision
    lists.  If this event happens, we refer to the block as being
    \emph{useless}.

    For a block $B_j$, let $Y_j$ be the total size of the collision
    lists of $\query$ for the projections of $B_j$ when ignoring good
    collisions altogether. We have that
    \begin{align*}
      \Ex{Y_j}%
      &\leq%
        \cardin{B_j} \cdot \cardin{\Pfar} \cdot 
        \cpi{\N}\pth{\bigl. (1+\eps) \rr}%
        \leq%
        \cardin{B_j}  n  \exp\pth{ - (1+\eps) \rr\N /\rr  } %
        =%
      \\&
          \cardin{B_j}n \exp\pth{ - (1+\eps) \N  }%
          =%
          \cardin{B_j}%
          =%
          O\pth{n^{1/(1+\eps)}},     
    \end{align*}
    since $\N = (\ln n )/(1+\eps)$. In particular, the $j$\th block is
    \emph{heavy}, if $Y_i \geq 10 \cardin{B_i}$. The probability for a
    block to be heavy, is $\leq 1/10$, by Markov's inequality.

    In particular, the probability that a block is heavy or useless,
    is at most $1/e + 1/10 \leq 1/2$. As such, with high probability,
    there is a light and useful block. Since the algorithm scans the
    blocks by their collision lists size, it follows that with high
    probability, the algorithm scans only light blocks before it stops
    the scanning, which is caused by getting to a point that belongs
    to $\Pclose$. As such, the query time of the algorithm is
    $O(d n^{1/(1+\eps)} \log n)$.
\end{proof}

Next, we analyze the data-dependent running time of the algorithm.
\begin{lemma}
    Let $\ell_i = \distH{\query}{\pnt_i}$, for $i=1,\ldots, n$, where
    $\PntSet = \brc{ \pnt_1, \ldots, \pnt_n} \subseteq \Hd$.  Let
    $\Delta$ be the smallest value such that
    $\sum_{i=1}^n \exp\pth{-\Delta \ell_i/\rr} \leq 1$. Then, the
    expected query time of the algorithm is
    $O\pth{ d e^{\Delta} \log n}$.
\end{lemma}
\begin{proof}
    The above condition implies that
    $\sCi{j} \leq \sCi{\Delta} \leq 1$, for any $j \geq \Delta$.  By
    \Eqref{collisions}, for $j \geq \Delta$, we have that
    \begin{math}
        \nCollX{j} = \Ex{X_j} = \sCi{j} \Ni{j} \leq \Ni{j}.
    \end{math}
    Thus, by Markov's inequality, with probability at least
    $1-1/\constB$, we have that $X_j \leq \constB \Ni{j}$, and the
    algorithm would terminate in this iteration.  As such, let $Y_j$
    be an indicator variable that is one of the algorithm reached the
    $j$\th iteration. However, for that to happen, the algorithm has
    to fail in iterations $\Delta, \Delta+1, \ldots, j-1$.  As such,
    we have that
    \begin{math}
        \Prob{Y_j = 1} = 1/\constB^{j-\Delta}.
    \end{math}

    The $j$\th iteration of the algorithm, if it happens, takes
    $O( d \Ni{j})$ time, and as such, the overall expected running
    time is proportional to
    \begin{math}
        \Bigl.\Ex{ \bigl. \smash{ d\sum_{j=1}^\N Y_j \Ni{j}}}.
    \end{math}
    Namely, the expected running time is bounded by
    \begin{align*}
      O\pth{ d\sum_{j=1}^\Delta \Ni{j} 
      + d \sum_{j=\Delta+1}^\N  
      \frac{\Ni{j}}{\pth{\constB}^{j - \Delta}}
      }%
      =%
      O\pth{d e^\Delta \log n + 
      d e^{\Delta} \sum_{j=\Delta+1}^\N  
      \frac{e^{j-\Delta} \log n}{\pth{\constB}^{j - \Delta}}
      }%
      =%
      O\pth{ d e^{\Delta} \log n},
    \end{align*}
    using the bound $\Ni{j} = O\pth{ e^{j} \log n} $ from
    \lemref{N:i}, and since $\constB > e$.
\end{proof}

\subsection{The result}

\begin{theorem}
    \thmlab{L:S:H:sensitive}%
    Given a set $\PntSet \subseteq \Hd = \brc{0,1}^d$ of $n$ points,
    and parameters $\eps >0$ and $\rr$, one can preprocess the point
    set, in $O\pth{ dn^{1+1/(1+\eps)} \log n}$ time and space, such
    that given a query point $\query \in \Hd$, one can decide
    (approximately) if $\distH{\query}{\PntSet} \leq \rr$, in
    $O\pth{ dn^{1/(1+\eps)} \log n}$ expected time, where $\distHC$ is
    the Hamming distance. Formally, the data-structure returns,
    either: \smallskip%
    \begin{compactitem}
        \item ``$\distH{\query}{\PntSet} \leq (1+\eps) \rr$'', and the
        data-structure returns a witness $\pnt \in \PntSet$, such that
        $\distH{\query}{\pnt} \leq (1+\eps) \rr$. This is the result
        returned if $\distH{\query}{\pnt} \leq \rr$.
        
        \smallskip%
        \item ``$\distH{\query}{\PntSet} > (1+\eps) \rr$'', and this
        is the result returned if
        $\distH{\query}{\PntSet} > (1+\eps) \rr$.
    \end{compactitem}
    \smallskip%
    The data-structure is allowed to return either answer if
    $\rr \leq \distH{\query}{\PntSet} \leq (1+\eps)\rr$.  The query
    returns the correct answer, with high probability.
    
    Furthermore, if the query is ``easy'', the data-structure would
    return the \emph{exact} nearest neighbor. Specifically, if
    $\distH{\query}{\PntSet} \leq \rr$, and there exists
    $\Delta < (\ln n )/ (1+\eps)$, such that
    $\sum_{\pnt \in \PntSet} \exp\pth{-\distH{\query}{\pnt}
       \Delta/\rr} \leq 1$,
    then the data-structure would return the exact nearest-neighbor in
    $\Bigl.O\pth{\bigl. d \exp(\Delta) \log n}$ expected time.
\end{theorem}

\begin{remark}
    \remlab{l:dim}%
    If the data is $k$ dimensional, in the sense of having bounded
    growth (see \exmref{low:dim}), then the above data-structure
    solves approximate \LSH in $O( d \log n)$ time, where the constant
    hidden in the $O$ depends (exponentially) on the data dimension
    $k$.

    This result is known, see Datar \etal~\cite[Appendix
    A]{diim-lshsb-04}. However, our data-structure is more general, as
    it handles this case with no modification, while the
    data-structure of Datar \etal is specialized for this case.
\end{remark}

\begin{remark}
    \remlab{f:tune}%
    Fine tuning the \LSH scheme to the hardness of the given data is
    not a new idea. In particular, Andoni \etal~\cite[Section
    4.3.1]{adiim-lshus-06} suggest fine tuning the \LSH construction
    parameters for the set of queries, to optimize the overall query
    time.

    Contrast this with the new data-structure of
    \thmref{L:S:H:sensitive}, which, conceptually, adapts the
    parameters on the fly during the query process, depending on how
    hard the query is.
\end{remark}



\section{Conclusions}

Ultimately, our data-structure is a prisoner of our underlying
technique of sampling coordinates. Thus, the main challenge is to come
up with a different approach that does not necessarily rely on such an
idea. In particular, our current technique does not work well for
points that are sparse, and have only few non-zero coordinates. We
believe that this problem should provide fertile ground for further
research.

\paragraph{Acknowledgments.}
The authors thank Piotr Indyk for insightful discussions about the
problem and also for the helpful comments on the presentation of this
paper.  The authors also thank \si{Jen} Gong, Stefanie Jegelka, and
Amin Sadeghi for useful discussions on the applications of this
problem.



\newcommand{\etalchar}[1]{$^{#1}$}
 \providecommand{\CNFX}[1]{ {\em{\textrm{(#1)}}}}
  \providecommand{\tildegen}{{\protect\raisebox{-0.1cm}{\symbol{'176}\hspace{-0.03cm}}}}
  \providecommand{\SarielWWWPapersAddr}{http://sarielhp.org/p/}
  \providecommand{\SarielWWWPapers}{http://sarielhp.org/p/}
  \providecommand{\urlSarielPaper}[1]{\href{\SarielWWWPapersAddr/#1}{\SarielWWWPapers{}/#1}}
  \providecommand{\Badoiu}{B\u{a}doiu}
  \providecommand{\Barany}{B{\'a}r{\'a}ny}
  \providecommand{\Bronimman}{Br{\"o}nnimann}  \providecommand{\Erdos}{Erd{\H
  o}s}  \providecommand{\Gartner}{G{\"a}rtner}
  \providecommand{\Matousek}{Matou{\v s}ek}
  \providecommand{\Merigot}{M{\'{}e}rigot}
  \providecommand{\Hastad}{H\r{a}stad\xspace}
  \providecommand{\CNFCCCG}{\CNFX{CCCG}}
  \providecommand{\CNFBROADNETS}{\CNFX{BROADNETS}}
  \providecommand{\CNFESA}{\CNFX{ESA}}
  \providecommand{\CNFFSTTCS}{\CNFX{FSTTCS}}
  \providecommand{\CNFIJCAI}{\CNFX{IJCAI}}
  \providecommand{\CNFINFOCOM}{\CNFX{INFOCOM}}
  \providecommand{\CNFIPCO}{\CNFX{IPCO}}
  \providecommand{\CNFISAAC}{\CNFX{ISAAC}}
  \providecommand{\CNFLICS}{\CNFX{LICS}}
  \providecommand{\CNFPODS}{\CNFX{PODS}}
  \providecommand{\CNFSWAT}{\CNFX{SWAT}}
  \providecommand{\CNFWADS}{\CNFX{WADS}}



\appendix

\section{For the $L_1$-case: Trading density for proximity}
\apndlab{eps:approx}%

The tail of a point in \lemref{lp:heavy:tail} has to be quite heavy,
for the data-structure to reject it as an \ANN. It is thus natural to
ask if one can do better, that is -- classify a far point as far, even
if the threshold for being far is much smaller (i.e., ultimately a
factor of $1+\eps$).  Maybe surprisingly, this can be done, but it
requires that such a far point would be quite dense, and we show how
to do so here. For the sake of simplicity of exposition the result of
this section is provided only under the $L_1$ norm.

\smallskip The algorithm is the same as the one presented in
\secref{sec:lp:alg}, except that for the given parameter $\epsilon$ we
use $(1+O(\epsilon))$-\ANN data-structures. We will specify it more
precisely at the end of this section. Also the total number of \ANN
data-structures is 
\begin{math}
    L = O\pth{ n^\delta \cdot \log \frac{n}{\eps} }.
\end{math}

\subsection{A tail of two points}

We start by proving a more nuanced version of \lemref{lp:heavy:block}.

\begin{lemma}
    \lemlab{heavy:block:eps}%
    Let $\xi, \eps \in (0,1)$ be parameters, and let $\pnt$ be a point
    in $\Re^d$ that is $( \xi \rr,(1+\eps)\rr)$-heavy.  Then
    \begin{math}
        \ds%
        \Prob{\Bigl. \norm{1}{ \subseq \pnt} \geq%
           (1+\eps/4) \pr \rr }%
        \geq%
        1 - \frac{{ 9\xi }}{\eps^2\pr },%
    \end{math}
\end{lemma}

\begin{proof}
    Let $\threshold = \xi \rr$ be a parameter and consider the
    $\threshold$-truncated point $\pntA = \truncY{\pnt}{\threshold}$.
    Since $\pnt$ is $( \xi \rr, (1+\eps)\rr)$-heavy, we have that
    $\norm{1}{\pnt} \geq \norm{1}{\pntA} \geq (1+\eps)\rr$.  Now, we
    have
    \begin{align*}
      \mu = \Ex{\bigl. \norm{1}{ \subseq \pntA} }%
      =%
      \pr \norm{1}{\pntA} %
      \geq (1+\eps) \pr \rr%
      \qquad\text{and}\qquad%
      \sigma^2%
      = %
      \Var{\Bigl. \norm{1}{ \subseq \pntA} }%
      =%
      \pr(1-\pr) \norm{2}{\pntA}^2.%
    \end{align*}
    Now, by Chebyshev's inequality, we have that
    \begin{align*}
      & \hspace{-1cm}%
        \Prob{ \Bigl. \norm{1}{ \subseq \pnt} \geq (1+\eps/4) \pr
        \rr}%
        \geq%
        \Prob{ \Bigl. \norm{1}{ \subseq \pntA} \geq (1+\eps/4) \pr
        \rr}%
        \geq%
        \Prob{ \Bigl. \norm{1}{ \subseq \pntA} \geq \frac{(1+\eps) \pr
        \rr}{1+\eps/2} }%
      \\&%
          \geq%
          \Prob{ \norm{1}{ \subseq \pntA} \geq \frac{\mu}{1+\eps/2} }%
          =%
          \Prob{ \norm{1}{ \subseq \pntA} - \mu \geq
          \pth{\frac{1}{1+\eps/2} -1 } \mu }%
      \\&%
          =%
          \Prob{ - \norm{1}{ \subseq \pntA} + \mu \leq
          {\frac{\eps/2}{1+\eps/2} \mu }}%
          =%
          1- \Prob{ - \norm{1}{ \subseq \pntA} + \mu \geq
          {\frac{\eps}{2+\eps} \mu }}%
      \\&%
          \geq%
          1 - \Prob{ \ts\cardin{\Bigl.\norm{1}{ \subseq \pntA} - \mu}
          \geq \frac{\eps}{2+\eps} \cdot \frac{ \mu}{\sigma} \cdot
          \sigma }%
          \geq%
          1 - \pth{\frac{2+\eps}{\eps}}^2 \pth{\frac{\sigma}{\mu}}^2%
      \\&%
          =%
          1 - \frac{9}{\eps^2} \cdot \frac{{ \pr(1-\pr)
          \norm{2}{\pntA}^2%
          }}{\pr^2 \norm{1}{\pntA}^2 }%
          \geq%
          1 - \frac{9}{\eps^2}\frac{{ (1-\pr) \xi
          }}{\pr (1+\eps) }%
          \geq%
          1 - \frac{ 9\xi }{\eps^2\pr },%
    \end{align*}
    by \lemref{lp:t:heavy}. Now by setting $\xi \leq \eps^3\pr/90$,
    this probability would be at least $1-\eps/10$.
\end{proof}

Similar to the \lemref{lp:light:good} by Markov's inequality we have
the following lemma.
\begin{lemma}
    \lemlab{light:good:eps}%
    Consider a point $\pnt$ such that
    $\snorm{1}{\tailY{\pnt}{k}}\leq \rr$. Conditioned on the event of
    \lemref{success}, we have that
    $\Prob{ \Bigl. \norm{1}{ \seq \pnt} \leq (1+\eps/32) \tTimes \pr
       \rr } \geq 1- \frac{1}{1+\eps/32} \geq \eps/33$,
    where $\seq \in \DistD{\pr}^\tTimes$.
\end{lemma}
\begin{lemma}
    \lemlab{heavy:far:eps}%
    Let $\RRA = (1+\eps/16) \tTimes \pr \rr$.  If $\pnt$ is
    $(\xi\rr,\RR)$-heavy point, then
    $\Prob{\Bigl. \norm{1}{ \seq \pnt} \geq \RRA } \geq 1 - {2
       n^{-\cTimes \eps^2/256}}$,
    assuming $\RR \geq (1+\eps) \rr$ and
    $\xi \leq \frac{\eps^3}{90(\cCoord k)}$.
\end{lemma}
\begin{proof}%
    For all $i$, let $Y_i = \norm{1}{ \subseq_i \pnt}$. By
    \lemref{heavy:block:eps}, with probability at least $(1-\eps/10)$,
    we have that
    \begin{math}
        Y_i %
        \geq (1+\eps/4)\pr \rr.
    \end{math}
    In particular, let $Z_i = \min( Y_i, (1+\eps/4)\pr \rr)$. We have
    that
    \begin{align*}
      \mu = \Ex{Z}%
      =%
      \Ex{\sum_{i=1}^\tTimes Z_i }%
      \geq%
      (1+\eps/4)\tTimes \pr \rr (1-\eps/10) \geq
      (1+\eps/8)\tTimes \pr\rr.
    \end{align*}
    Now, by Hoeffding's inequality, we have that
    \begin{align*}
      \Prob{\Bigl. \norm{1}{ \seq \pnt} \leq \RRA}%
      &\leq%
        \Prob{\Bigl. Z \leq \RRA}%
        \leq%
        \Prob{ \Bigl. \cardin{ Z - \mu} \geq \frac{\eps}{16} \tTimes
        \pr \rr}%
        \leq%
        2 \exp\pth{ - \frac{2(\tTimes \pr \rr \eps/16)^2 }{\tTimes
        \pth{ (1+\eps/4)\pr \rr}^2}}%
      \\&%
          \leq%
          2 \exp\pth{ - \frac{t\eps^2}{16^2}}%
          \leq%
          \frac{2}{n^{\cTimes\eps^2/256}}.
    \end{align*}
\end{proof}
\begin{lemma}
    \lemlab{heavy:tail:eps}
    Let $\delta , \eps \in (0,1)$ be two parameters. For a query point
    $\query \in \Re^d$, let $\nnpnt \in \PntSet$ be its $k$-fold
    nearest neighbor in $\PntSet$, and let
    $\rr = \snorm{1}{\tailY{(\nnpnt-\query)}{k}}$. Then, with high
    probability, the algorithm returns a point $\pntA \in \PntSet$,
    such that $\query - \pntA$ is a
    $\pth{\xi\rr, \rr(1+2\eps)}$-light, where
    $\xi = O(\delta\eps^5/k)$. The data-structure performs
    $O\pth{n^{\delta} \cdot \frac{\log n}{\eps}}$ of $(1+\eps/64)$-\ANN
    queries.
\end{lemma}

\begin{proof}
    As before, we set $\cTimes = 512/\eps^2$ and
    $\cCoord = \cTimes / \cDSTimes$. Also, by conditions of
    \lemref{heavy:far:eps} we have
    $\xi\leq\eps^3/(90 \cdot\cCoord k) = \eps^3\cDSTimes/(90 \cTimes
    k) = \eps^5\cDSTimes/(90\cdot 512 \cdot k)$.
    Also, let $L = (n^{\cDSTimes}\cdot \frac{\log n}{\eps})$.  Let $S$
    be the set of $k$ largest coordinates in
    $\pntB = \query - \nnpnt$. By similar arguments as in
    \lemref{lp:heavy:tail}, there exists
    $m = \Omega(\frac{\log n}{ \eps})$ data structures say
    $ \DA_1, \ldots, \DA_m$ such that $\seq_i \in \DistD{\pr}^\tTimes$
    does not contain any of the coordinates of $S$.

    Since by assumption
    $\snorm{1}{\rmC{\pntB}{S}} =\snorm{1}{\tailY{\pntB}{k}} \leq r$,
    and by \lemref{light:good:eps}, with probability at least
    $\eps/33$ the distance of $\seq_i \nnpnt$ from $\seq_i \query$ is
    at most $(1+\eps/32) \tTimes \pr \rr$. Since there are
    $\Omega( \frac{\log n}{\eps})$ such structures, we know that, with
    high probability, for one of them, say $\DA_1$, this holds.  By
    \lemref{heavy:far:eps}, any point $\rmC{\pnt}{S}$ (of $\PntSet$)
    that is $(\xi\rr, (1+\eps) \rr)$-heavy, would be in distance at
    least $(1+\eps/16) \tTimes \pr \rr$ in the projection $\seq_1$
    from the projected $\query$, and since $\DA_1$ is a
    $(1+\eps/64)$-\ANN data-structure under the $L_1$ norm, we
    conclude that no such point can be returned.  Note that since for
    the reported point $\pntA$, the point $\rmC{\pntA}{S}$ cannot be
    $(\xi\rr,\rr(1+\eps))$-heavy, and the coordinates in $S$ can
    contribute at most $k\xi\rr\leq \cDSTimes\eps^5\rr \leq \eps \rr$,
    the point $\pntA$ cannot be $(\xi\rr ,\rr(1+2\eps))$-heavy. Thus,
    the data-structure returns the desired point with high
    probability.

    As for the query performance, the data-structure performs
    $\DSTimes$ queries of $(1+\eps/64)$-\ANN data-structures.
\end{proof}

By \lemref{lp:light}, we get the following corollary.
\begin{corollary}
    Given a query point $\query \in \Re^d$, let $\nnpnt \in \PntSet$
    be its $k$-robust nearest neighbor in $\PntSet$, and
    $\rr = \snorm{1}{\tailY{(\nnpnt-\query)}{k}}$. Then, with high
    probability, the algorithm returns a point $\pntA \in \PntSet$,
    such that
    \begin{math}
        \norm{1}{\tailY{(\query -
              \pntA)}{O(\frac{k}{\eps^5\cDSTimes})}}\leq \rr
        (1+2\eps).
    \end{math}
\end{corollary}





\end{document}